\documentclass{CSML}
\pdfoutput=1

\usepackage{lastpage}
\newcommand{\dOi}{14(2:15)2018}
\lmcsheading{}{1--\pageref{LastPage}}{}{}%
{Feb.~06, 2014}{Jun.~19, 2018}{}

\usepackage{hyperref}
\hypersetup{hidelinks}

\usepackage{graphicx}
\usepackage{tikz}
\usepackage{bm}
\usepackage{amsmath}
\usepackage{amssymb}

\usepackage{pseudocode}
\usepackage{chngcntr}

\counterwithout{pseudocode}{section}

\newcounter{newpseudonum}[pseudocode]

\ifpdf
  \providecommand{\refline}[1]{\hyperref[#1]{(\ref*{#1})}}
\else
  \providecommand{\refline}[1]{\ref*{#1}}
\fi

\renewcommand{\RETURN}[1]{\ifthenelse{\equal{#1}{} }{\mbox{\bfseries return}}{\mbox{\bfseries return}#1}}
\newcommand{\FUNCTION}[2]{\mbox{\bfseries proc }\mbox{\textsc{#1}}\left(\ensuremath{#2}\right)\\}

\newlength{\pcodewidth}
\newenvironment{code}[1]{
\begin{Sbox}
\!\!\begin{minipage}{#1}
\bfseries
\noindent
\scriptsize
$$
\begin{array}{@{\hspace*{1ex}}lr@{}}
}{
\end{array}
$$
\end{minipage}\vspace{-2mm}
\end{Sbox}
\shadowbox{\TheSbox}{}
}

\usetikzlibrary{decorations,arrows,shapes}
\usetikzlibrary{positioning}
\usetikzlibrary{arrows,automata}
\usepackage{url}
\DeclareMathAlphabet{\mathpzc}{OT1}{pzc}{m}{it}

\newcommand{\ii}[2]{{\mathpzc{#1}\mathpzc{#2}}_{ii}}
\newcommand{\ip}[1]{{\mathpzc{#1}}_{ip}}
\renewcommand{\pi}[2]{{\mathpzc{#1}\mathpzc{#2}}_{pi}}
\newcommand{\pp}[1]{{\mathpzc{#1}}_{pp}}

\newcommand{\Lin}{\mathsf{Lin}}
\newcommand{\Den}{\mathsf{Den}}
\newcommand{\Dis}{\mathsf{Dis}}
\newcommand{\Unb}{\mathsf{Unb}}
\newcommand{\Fin}{\mathsf{Fin}}
\newcommand{\allr}{\mathfrak{R}}
\newcommand{\alli}{\mathfrak{I}}
\newcommand{\allm}{\mathfrak{M}}
\newcommand{\allp}{\mathfrak{P}}

\begin{document}

\title[A Theory of Points and Intervals (I)]{An Integrated First-Order Theory of Points and Intervals over Linear Orders (Part I)}

\author[Willem Conradie et al.]{Willem Conradie}	
\address{School of Mathematics, University of the Witwatersrand\\
        Johannesburg, South Africa}	
\email{willem.conradie@wits.ac.za}  

\author[]{Salih Durhan}	
\address{Nesin Mathematics Village \\
         Sirince, Turkey}	
\email{salih1@gmail.com}  

\author[]{Guido Sciavicco}	
\address{Department of Mathematics and Computer Science\\
        University of Ferrara\\
        Ferrara, Italy}	
\email{guido.sciavicco@unife.it}  
\thanks{We would like to thank Dr. Davide Bresolin, of the University of Padova, for his help.
  The research of the first author was supported by the National Research Foundation of South Africa grant number 81309..
}	 



\keywords{Interval-based Temporal Logics; Expressiveness; Allen's Relations}
\subjclass{[Theory of computation]: Logic --- Modal and temporal logics; [Computing methodologies]: Artificial intelligence --- Knowledge representation and reasoning ---  Temporal reasoning}


\begin{abstract}
  There are two natural and well-studied approaches to temporal ontology and reasoning: point-based and interval-based. Usually, interval-based temporal reasoning deals with points as a particular case of duration-less intervals. A recent result by Balbiani, Goranko, and Sciavicco presented an explicit two-sorted point-interval temporal framework in which time instants (points) and time periods (intervals) are considered on a par, allowing the perspective to shift between these within the formal discourse. We consider here two-sorted first-order languages based on the same principle, and therefore including relations, as first studied by Reich, among others, between points, between intervals, and inter-sort. We give complete classifications of its sub-languages in terms of relative expressive power, thus determining how many, and which, are the intrinsically different extensions of two-sorted first-order logic with one or more such relations. This approach roots out the classical problem of whether or not points should be included in a interval-based semantics.
\end{abstract}

\maketitle

\section{Introduction}\label{sec:intro}

The relevance of temporal logics in many theoretical and applied areas of computer science and AI, such as theories of action and change, natural language analysis and processing, and constraint satisfaction problems, is widely recognized. While the predominant approach in the study of temporal reasoning and logics has been based on the assumption that time points (instants) are the primary temporal ontological entities, there has also been significant activity in the study of interval-based temporal reasoning and logics over the past two decades. The variety of binary relations between intervals in linear orders was first studied systematically by Allen~\cite{Allen87,Allen:1983:MKA,JLOGC::AllenF1994}, who explored their use in systems for time management and planning. Allen's work and much that follows from it is based on the assumption that time can be represented as a dense line, and that points are excluded from the semantics. At the modal level, Halpern and Shoham~\cite{JACM::HalpernS1991} introduced the multi-modal logic HS that comprises modal operators for all possible relations (known as Allen's relations~\cite{Allen:1983:MKA}) between two intervals in a linear order, and it has been followed by a series of publications studying the expressiveness and decidability/undecidability and complexity of the fragments of HS, e.g.,~\cite{tcs2014,amai2014}. Many studies on interval logics have considered the so-called `non-strict' interval semantics, allowing point-intervals (with coinciding endpoints) along with proper ones, and thus encompassing the instant-based approach, too; more recent ones, instead, started to treat pure intervals only. Yet, little has been done so far on the formal treatment of both temporal primitives, points and intervals, in a unified two-sorted framework. A detailed philosophical study of both approaches, point-based and interval-based, can be found in~\cite{vB91} (see also~\cite{590368}). A similar mixed approach has been studied in~\cite{DBLP:journals/ci/AllenH89}. ~\cite{DBLP:journals/cj/MaH06} contains a study of the two sorts and the relations between them in dense linear orders. More recently, a modal logic that includes different operators for points and interval has been presented in~\cite{DBLP:journals/entcs/BalbianiGS11}.

\medskip

The present paper provides a systematic treatment of point and interval relations (including equality between points and between intervals treated on the same footing as the other relations) at the first-order level. Our work is motivated, among other observations, by the fact that natural languages incorporate both ontologies on a par, without assuming the primacy of one over the other, and have the capacity to shift the perspective smoothly from instants to intervals and vice versa within the same discourse, e.g.: {\em when the alarm goes on, it stays on until the code is entered}, which contains two instantaneous events and a non-instantaneous one. Moreover, there are various temporal scenarios which neither of the two ontologies alone can grasp properly since neither the treatment of intervals as the sets of their internal points, nor the treatment of points as `instantaneous' intervals, is really adequate. The technical identification of intervals with sets of their internal points, or of points as instantaneous intervals leads also to conceptual problems like the confusion of events and fluents. Instantaneous events are represented by time intervals and should be distinguished from instantaneous holding of fluents, which are evaluated at time points: therefore, the point $a$ should be distinguished from the interval $[a,a]$, and the truths in these should not necessarily imply each other. Finally, we note that, while differences in expressiveness have been found between the strict and non-strict semantics for some interval logics (see~\cite{ijcai11}, for example), so far, no distinction in the decidability of the satisfiability has been found. Therefore, we believe that an attempt to systemize the role of points, intervals, and their interaction, would make good sense not only from a purely ontological point of view, but also from algorithmic and computational perspectives.

\medskip

\noindent{\bf Previous Work and Motivations.} As presented in the early work of van Benthem~\cite{vB91} and Allen and Hayes~\cite{Allen85}, interval temporal reasoning can be formalized as an extension of first-order logic with equality with one or more relations, and the properties of the resulting language can be studied; obviously, the same applies when relations between points are considered too. In this paper we ask the question: interpreted over linear orders, how many and which expressively different languages can be obtained by enriching first-order logic with relations between intervals,  between points, and between intervals and points? Since, as we shall see, there are 26 different relations (including equality of both sorts) between points, intervals, and points and intervals, $2^{26}$ is an upper bound on this number. (It is worth noticing that in~\cite{DBLP:journals/cj/MaH06} the authors distinguish 30 relations, instead of 26; this is due to the fact that the concepts of the point $a$ {\em starting} the interval $[a,b]$ and {\em meeting} it are considered to be different.) However, since certain relations are definable in terms of other ones, the actual number is less and in fact, as we shall show, much less. The answer also depends on our choices of certain semantic parameters, specifically, the class of linear orders over which we construct our interval structures. In this paper, in Part I, we consider the classification problem relative to:

\begin{enumerate}[label={(\emph{\roman*})}]
\item the class of all linear orders;
\item the class of all  weakly discrete linear orders (i.e., orders in which every point with a successor/predecessor has an immediate one).
\end{enumerate}

\noindent In Part II of this paper we consider:
\begin{enumerate}[resume*]
\item the class of all dense linear orders;
\item the class of all unbounded linear orders;
\end{enumerate}

\medskip

Apart from the intrinsic interest and naturalness of this classification problem, its outcome has some important repercussions, principally in the reduction of the number of cases that need to be considered in other problems relating to these languages. For example, it reduces the number of representation theorems that are needed: given the {\em dual} nature of time intervals (i.e., they can be abstract first-order individuals with specific characteristics, or they can be defined as ordered pairs over a linear order), one of the most important problems that arises is the existence or not of a {\em representation theorem}. Consider any class of linear orders: given a specific extension of first-order logic with a set of interval relations (such as, for example, {\em meets} and {\em during}), does there exist a set of axioms in this language which would constrain (abstract) models of this signature to be isomorphic to concrete ones?  Various representation theorems exist in the literature for languages that include interval relations only: van Benthem~\cite{vB91}, over rationals and with the interval relations {\em during} and {\em before}, Allen and Hayes~\cite{Allen85}, for the dense unbounded case without point intervals and for the relation {\em meets}, Ladkin~\cite{Ladkin}, for point-based structures with a quaternary relation that encodes meeting of two intervals, Venema~\cite{JLOGC::Venema1991}, for structures with the relations {\em starts} and {\em finishes}, Goranko, Montanari, and Sciavicco~\cite{GMS03}, for linear structures with {\em meets} and {\em met-by}, Bochman~\cite{DBLP:journals/ndjfl/Bochman90}, for point-interval structures, and Coetzee~\cite{Coetzee:MThesis} for dense structure with {\em overlaps} and {\em meets}. Clearly, if two sets of relations give rise to expressively equivalent languages, two separate representations theorems for them are not needed.  In which cases are representation theorems still outstanding? Preliminary works that provide similar classifications appeared in~\cite{DBLP:conf/caepia/ConradieS11} for first-order languages with equality and only interval-interval relations, and in~\cite{time2012} for points and intervals (with equality between intervals treated on a par with the other relations) but only over the class of all linear orders. Finally, a complete study of first-order interval temporal logics enables a deeper understanding of their modal counterparts based on their shared relational semantics.

\medskip

\noindent{\bf Structure of the paper.} This paper is structured as follows. Section~\ref{sec:basics} provides the necessary preliminaries, along with an overview of the general methodology used in this paper. In Section~\ref{sec:lin} we study the expressive power of the language by analyzing the definability
properties of each basic relation in the class $\Lin$. Section~\ref{sec:lin-incomp} deals with the `other half' of the picture, that is, undefinability results, and presenting all maximally incomplete sets in this class (i.e., those subsets of relations that do not allow one to define the remaining ones, and are maximal in this sense with respect to the subset relation); we also deal with completeness and incompleteness results for the class of all discrete linear orders in this section. Section~\ref{sec:harvest} presents an account of all our results in a structured way, including the projections of these to natural sub-languages, before concluding. The classes of all dense and the class of all unbounded linearly ordered sets will be treated in Part II (forthcoming).

\section{Basics}\label{sec:basics}

\subsection{Syntax and semantics}

\begin{table}[t]
\centering
\begin{tikzpicture}[scale=1]
\tikzstyle{every node}=[font=\small]
\draw (0,0)node(op){};
%

\draw (op) ++(0,-.75)node(meets){};
\draw (meets)++(0,-.75)node(later){};
\draw (meets)++(0,-1.5)node(starts){};
\draw (meets)++(0,-2.25)node(finishes){};
\draw (meets)++(0,-3)node(during){};
\draw (meets)++(0,-3.75)node(overlaps){};
\draw (meets)++(-1.0,0)node[right](Ra){$(\mbox{{\em meets,m}})\,\,[a,b]~\ii{3}{4}~[c,d] \Leftrightarrow b=c$};
\draw (meets)++(-1.0,-.75)node[right](Rl){$(\mbox{{\em before,b}})\,\,[a,b]~\ii{4}{4}~[c,d] \Leftrightarrow b < c$};
\draw (meets)++(-1.0,-1.5)node[right](Rs){$(\mbox{{\em starts,s}})\,\,[c,d]~\ii{1}{4}~[a,b] \Leftrightarrow a=c, d < b$};
\draw (meets)++(-1.0,-2.25)node[right](Rf){$(\mbox{{\em finishes,f}})\,\,[c,d]~\ii{0}{3}~[a,b] \Leftrightarrow b=d, a < c$};
\draw (meets)++(-1.0,-3)node[right](Rd){$(\mbox{{\em during,d}})\,\,[c,d]~\ii{0}{4}~[a,b] \Leftrightarrow a < c, d < b$};
\draw (meets)++(-1.0,-3.75)node[right](Ro){$(\mbox{{\em overlaps,o}})\,\,[a,b]~\ii{2}{4}~[c,d] \Leftrightarrow a < c < b < d$};
\draw[red,|-|] (meets) ++(7.6,.75)node[above](a){$a$} -- ++(2,0)node[above](b){$b$};
\draw[dashed,red,help lines,thick] (a) -- ++(0,-5.25);
\draw[dashed,red,help lines,thick] (b) -- ++(0,-5.25);
\draw[|-|] (b) ++(0,-1) ++(0,0)node[above](Ac){$c$}
-- ++(1,0)node[above](Ad){$d$};
\draw[|-|] (b) ++(0,-1) ++(.5,-.75)node[above](Lc){$c$}
-- ++(1,0)node[above](Ld){$d$};
\draw[|-|] (a) ++(0,-1) ++(0,-1.5)node[above](Bc){$c$}
-- ++(.5,0)node[above](Bd){$d$};
\draw[|-|] (b) ++(0,-1) ++(-.5,-2.25)node[above](Ec){$c$}
-- ++(.5,0)node[above](Ed){$d$};
\draw[|-|] (a) ++(0,-1) ++(.5,-3)node[above](Dc){$c$}
-- ++(1,0)node[above](Dd){$d$};
\draw[|-|] (a) ++(0,-1) ++(1,-3.75)node[above](Oc){$c$}
-- ++(2,0)node[above](Od){$d$};
\end{tikzpicture}
\caption{Interval-interval relations, a.k.a. Allen's relations. The equality relation is not depicted.}
\label{ii:relations}
\end{table}

Given a linear order $\mathbb D=\langle D,<\rangle$, we call the elements of $D$ \emph{points} (denoted by $a,b,\ldots$) and define an {\em interval} as an ordered pair $[a,b]$ of points in $D$, where $a<b$. Abstract intervals will be denoted by $I,J,\ldots,$ and so on.
Now, as we have mentioned above, there are 13 possible relations, including equality, between any two intervals. From now on, we
call these {\em interval-interval} relations. Besides equality, there are 2 different relations that may hold between any two points ({\em before}
and {\em after}), called hereafter {\em point-point} relations, and 5 different relations that may hold between a point and an interval and vice-versa: we call those {\em interval-point} and {\em point-interval} relations, respectively, and we use the term {\em mixed} relations to refer to them indistinctly. Interval-interval relations are exactly Allen's relations~\cite{Allen:1983:MKA}; point-point relations are the classical relations on a linear order, and mixed relations will be explained below. Traditionally, interval relations are represented by the initial letter of the description of the relation, like $m$ for {\em meets}. However, when one considers more relations (like point-point and point-interval relations) this notation becomes confusing, and even more so in the presence of more relations, e.g. when one wants to consider interval relations over a {\em partial order}\footnote{This paper is focused on linear orders only; nevertheless, it is our intention to complete this study to include the treatment of partial orders also, and, at this stage, we want to make sure that we will be able to keep the notation consistent.}. We introduce the following notation to resolve this issue: an interval $[a,b]$ induces a
partition of $\mathbb D$ into five regions (see~\cite{lig91}): region 0 which contains all
points less than $a$, region 1 which contains $a$ only,
region 2 which contains all the points strictly between $a$ and $b$, region 3 which contains
only $b$ and region 4 which contains the points greater than $b$. Likewise, a point $c$ induces a partition of $\mathbb D$ into 3 pieces: region 0 contains all the points less than $c$, region 2 contains only $c$, and region 4 contains all the points greater than $c$. Interval-interval relations will be denoted by $I\ii{k}{k'}J$ (where the subscript $_{ii}$ refers to interval-interval relations), where $k,k'\in\{0,1,2,3,4\}$, and $k$ represent the region of the partition induced by $I$ in which the left endpoint of $J$ falls, while $k'$ is the region of the same partition in which the right endpoint of $J$ falls; for example, $I\ii{3}{4}J$ is exactly Allen's relation {\em meets}. Similarly, interval-point relations will be denoted by $I\ip{k\,}a$ (where the subscript $_{ip}$ stands for interval-point relations), where $k$ represents the position of $a$ with respect to $I$; for example, $I\ip{4}a$ is the relation {\em before}. Analogously, point-point relations will be denoted by the symbol $\pp{k\,}$, and point-interval relations by the symbol $\pi{k}{k'}$. For point-point relations it is more convenient to use $<$ instead of $\pp{4}$, and $>$ instead of $\pp{0}$.  In Tab.~\ref{ii:relations} we show six of the interval-interval relations, along with its original nomenclature and symbology, and in Tab.~\ref{ip:relations} we show the interval-point relations. Finally, we consider a equality per sort, using  $=_i$ to denote $\ii 13$ (equality between intervals), and $=_p$ to denote $\pp 2$ (the equality between points). Now, given any of the mentioned relations $r$, its inverse, generically denoted by $\bar r$, can be obtained by inverting the roles of the objects in the case of non-mixed relations; for example, the inverse of the relation $\ii{2}{2}$ (Allen's relation {\em contains}) is the relation $\ii{0}{4}$ (Allen's relation {\em during}). On the other hand, mixed relations present a different situation: the inverse of a point-interval relation is an interval-point relation; thus, for example, the inverse of $\ip{3}$ is $\pi{0}{2}$. Finally, notice that some combinations are forbidden: for instance, the relation $\pi{2}{2}$ makes no sense, as all intervals have a non-zero extension.

\begin{defi}
We shall denote by: $\allr$ the set of all above described relations; $\alli\subset\allr$ the subset of all 13 interval-interval relations (Allen's relations) including the relation $=_i$; $\allm\subset\allr$ the subset of all mixed relations; $\allp\subset\allr$ the subset of all point-point relations including the relation $=_p$. Clearly, $\allr=\alli \bigcup \allm \bigcup \allp$.
\end{defi}

\begin{table}[t]
\centering
\begin{picture}(100,100)
\color{red}
\put(60,90){\line(1,0){40}}
\put(60,87){\line(0,1){6}}
\put(100,87){\line(0,1){6}}
\put(58,94){\small{$a$}}
\put(98,94){\small{$b$}}
\color{black}
\put(-70,74){\small{$[a,b]~\ip{3}~c \Leftrightarrow b=c$}}
\put(98,77){$\cdot$}
\put(97,81){\small{$c$}}
\put(-70,58){\small{$[a,b]~\ip{4}~c \Leftrightarrow b<c$}}
\put(120,60){$\cdot$}
\put(119,64){\small{$c$}}
\put(-70,42){\small{$[a,b]~\ip{2}~c \Leftrightarrow a<c<b$}}
\put(80,45){$\cdot$}
\put(79,49){\small{$c$}}
\put(-70,28){\small{$[a,b]~\ip{1}~c \Leftrightarrow a=c$}}
\put(59,29){$\cdot$}
\put(58,35){\small{$c$}}
\put(-70,12){\small{$[a,b]~\ip{0}~c \Leftrightarrow c<a$}}
\put(40,10){$\cdot$}
\put(39,15){\small{$c$}}
\multiput(60,15)(0,8){9}{\line(0,1){3}}
\multiput(100,15)(0,8){9}{\line(0,1){3}}
\end{picture}
\caption{Interval-point relations.}
\label{ip:relations}
\end{table}

\begin{defi}
In the following, we denote by:

\begin{enumerate}[label={(\emph{\roman*})}]
\item $\Lin$ the class of all linear orders;
\item $\Den$ the class of all dense linear orders, that is, the class of all linear orders where there exists a point in between any two distinct points;
\item $\Dis$ the class of all weakly discrete linear orders, that is, the class of all linear orders where each point, other than the least (resp., greatest) point, if there is one, has a direct predecessor (resp., successor) -- by a {\em direct predecessor} of $a$ we of course mean a point $b$ such that $b < a$ and for all points $c$, if $c < a$ then $c \leq b$, and the notion of a {\em direct successor} is defined dually;
\item $\Unb$ the class of all unbounded linear orders, that is, the class of all linear order such that for every point $a$ there exists a point $b>a$ and a point $c<a$.
\end{enumerate}
\end{defi}

\begin{defi}
Given a linear order $\mathbb D$, and given the set $\mathbb I(\mathbb D)=\{[a,b]\mid a,b\in \mathbb D, a<b\}$ of all intervals built on $\mathbb D$:
\begin{itemize}
\item a \emph{concrete interval structure of signature $S$} is a relational structure $\mathcal{F} = \langle \mathbb I(\mathbb D), r_1,$ $r_2, \ldots,r_n \rangle$, where $S = \{r_1, \ldots,r_n \} \subseteq \alli$, and
\item a \emph{concrete point-interval structure of signature $S$} is a two-sorted relational structure $\mathcal{F} = \langle \mathbb D,\mathbb I(\mathbb D), r_1, r_2, \ldots,$ $r_n \rangle$, where $S = \{r_1, \ldots, r_n \} \subseteq \allr$.
\end{itemize}
\end{defi}

\noindent Since all relations between intervals, points, and all mixed relations are already implicit in $\mathbb I(\mathbb D)$, we shall often simply write $\langle \mathbb I(\mathbb D)\rangle$ for a concrete interval structure $\langle \mathbb I(\mathbb D), r_1, r_2, \ldots, r_n \rangle$, and $\langle \mathbb{D}, \mathbb I(\mathbb D)\rangle$ for a concrete point-interval structure $\langle \mathbb D,\mathbb I(\mathbb D), r_1, r_2, \ldots, r_n \rangle$; this is in accordance with the standard usage in much of the literature on interval temporal logics. Moreover, we denote by $FO+S$ the language of first-order logic without equality and relation symbols corresponding to the relations in $S$. Finally, $\mathcal{F}$ is further said to be \emph{of the class $\mathrm C$} ($\mathrm C\in\{\Lin,\Den,\Dis,\Unb,\Fin\}$) when $\mathbb D$ belongs to the specific class of linear orders $\mathrm C$.

%
%
%
%
%
%

\subsection{(Un)definability and Truth Preserving Relations}

We describe here the most important tools that we use to classify the expressive power of our (sub-)languages.

\begin{defi}
Let $S \subseteq \allr$, and $\mathrm C$ a class of linear orders. We say that $FO + S$ {\em defines} $r\in \allr$ over $\mathrm C$, denoted by $FO+S\rightarrow_{\mathrm C} r$, if there exists $FO+S$-formula $\varphi(x,y)$ such that $\varphi(x,y) \leftrightarrow r(x,y)$ is valid on the class of concrete point-interval structures of signature $(S \cup \{ r \})$ based on $\mathrm C$.
\end{defi}

\noindent By $FO+S\rightarrow r$ we denote the fact that $FO+S\rightarrow_{\Lin} r$ (and hence $FO+S\rightarrow_{\mathrm C} r$ for every $\mathrm C\in\{\Lin,\Den,\Dis,\Unb,\Fin\}$). Obviously, $FO+S\rightarrow r$ for all $r \in S$.

\begin{defi}
Let $S,S'\subseteq \allr$ and $\mathrm{C}$ a class of linear orders. We say that $S$ is:
\begin{itemize}
\item {\em $S'$-complete over $\mathrm C$} (resp., {\em $S'$-incomplete over $\mathrm C$}) if and only if $FO+S\rightarrow_{\mathrm C} r$ for all $r\in S'$ (resp., $FO+S\not\rightarrow_{\mathrm C} r$ for some $r\in S'$), and
\item {\em minimally $S'$-complete over $\mathrm C$} (resp., {\em maximally $S'$-incomplete over $\mathrm C$}) if and only if it is $S'$-complete (resp., $S'$-incomplete) over $\mathrm C$, and every proper subset (resp., every proper superset) of $S$ is $S'$-incomplete (resp., $S'$-complete) over the same class.
\end{itemize}
\end{defi}

\noindent The notion of (minimally) $r$-completeness and (maximally) $r$-incompleteness over $\mathrm C$ is immediately deduced from the above one, by taking $S'=\{r\}$ and denoting the latter simply by $r$. Moreover, one can project the above definitions over some interesting subsets of $\allr$, such as $\alli,\allm$ or $\allp$, obtaining relative completeness and incompleteness.

%


\medskip

Let $C' \subseteq C$ be two classes of linear orders. Notice that if $FO+S\rightarrow_{\mathrm C} r$ then $FO+S\rightarrow_{\mathrm C'} r$ and, contrapositively, that if $FO+S \not \rightarrow_{\mathrm C'} r$ then $FO+S \not \rightarrow_{\mathrm C} r$. So specifically, if $S$ is $S'$-complete over $\mathrm C$, then it is also $S'$-complete over $\mathrm C'$. Also, if $S$ is $S'$-incomplete over $\mathrm C'$, then it is also $S'$-incomplete over $\mathrm C$. Notice however, that minimality and maximality of complete and incomplete sets does not necessarily transfer between super and subclasses in a similar way. In what follows, in order to prove that $FO+S \not \rightarrow_{\mathrm C} r$ for some $r$ and some class $\mathrm C$, we shall repeatedly apply the following definition and (rather standard) procedure.

\begin{defi}\label{Def:S:Truth:Pres:Rel}
Let $\mathcal{F} = \langle \mathbb D,\mathbb I(\mathbb D),S\rangle$ and $\mathcal{F}' = \langle \mathbb D',\mathbb I(\mathbb D'), S\rangle$ be concrete structures where $S\subseteq\allr$. A binary relation $\zeta\subseteq(\mathbb D\cup\mathbb I(\mathbb D)) \times (\mathbb D' \cup \mathbb I(\mathbb D'))$ is called a \emph{surjective $S$-truth preserving relation} if and only if:

\medskip

\begin{enumerate}[label={(\emph{\roman*})}]
\item $\zeta$ respects sorts, i.e., $\zeta = \zeta_p \cup \zeta_i$, where $\zeta_p \subseteq \mathbb D \times \mathbb D'$ and $\zeta_i \subseteq \mathbb I(\mathbb D) \times \mathbb I(\mathbb D')$;
\item $\zeta$ respects the relations in $S$, i.e., if $(a,a'),(b,b')\in\zeta_p$ and $(I,I'),(J,J')\in\zeta_i$, then:
    \begin{enumerate}[nosep]
    \item $r(a,b)$ if and only if $r(a',b')$ for every point-point relation $r \in S$;
    \item $r(I,a)$ if and only if $r(I',a')$ for every interval-point relation $r\in S$;
    \item $r(I,J)$ if and only if $r(I',J')$ for every interval-interval relation $r \in S$;
    \end{enumerate}

\item $\zeta$ is total and surjective, i.e.:
    \begin{enumerate}[nosep]
    \item for every $a \in \mathbb D$ (resp., $I \in \mathbb I(\mathbb D)$), there exist $a' \in \mathbb D'$ (resp., $I' \in \mathbb I(\mathbb D')$) such that $(a,a')\in\zeta_p$ (resp., $(I,I')\in\zeta_i$);
    \item for every $a' \in \mathbb D'$ (resp., $I' \in \mathbb I(\mathbb D')$), there exist $a \in \mathbb D$ (resp., $I \in \mathbb I(\mathbb D)$) such that $(a,a')\in\zeta_p$ (resp., $(I,I')\in\zeta_i$).
    \end{enumerate}
\end{enumerate}
\end{defi}

\noindent If we add to Definition \ref{Def:S:Truth:Pres:Rel} the requirement that that $\zeta$ should be functional, we obtain nothing but the definition of an isomorphism between two-sorted first-order structures or, equivalently, an isomorphism between single sorted first-order structures with predicates added for `point' and `interval' (see e.g. \cite{hodges1993model}). As one would expect, surjective $S$-truth preserving relations preserve the truth of all first-order formulas in signature $S$. This is stated in Theorem \ref{theo:truth}, below. The reason why we consider only interval-point relations instead of all mixed relations is that, as we shall explain, we can limit ourselves to work without inverse relations, and point-interval relations are the inverse of interval-point ones.

\begin{defi}
If $\zeta$ is a surjective $S$-truth preserving relation, we say that $\zeta$ {\em breaks} $r\not\in S$ if and only if there are:
\begin{enumerate}[label={(\emph{\roman*})}]
    \item $(a,a'),(b,b')\in\zeta_p$ such that $r(a,b)$ but $\neg r(a',b')$, if $r$ is point-point, or
    \item $(a,a')\in\zeta_p$ and $(I,I')\in\zeta_i$ such that $r(I,a)$ but $\neg r(I',a')$, if $r$ is interval-point, or
    \item $(I,I'),(J,J')\in\zeta_i$ such that $r(I,J)$ but $\neg r(I',J')$, if $r$ is interval-interval.
\end{enumerate}
\end{defi}

\noindent The following result is, as already mentioned, a straightforward generalization of the classical result on the preservation of truth under isomorphism between first-order structures, and it is proved by an easy induction on formulas, using clause (ii) of Definition \ref{Def:S:Truth:Pres:Rel} to establish the base case for atomic formulas and clause (iii) for the inductive step for the quantifiers.

\begin{thm}\label{theo:truth}
If $\zeta=\zeta_p\cup\zeta_i$ is a surjective $S$-truth preserving relation between $\mathcal{F} = \langle \mathbb D,\mathbb I(\mathbb D),S\rangle$ and $\mathcal{F}' = \langle \mathbb D',\mathbb I(\mathbb D'),S\rangle$, and $a_1,\ldots,a_k \in \mathbb D$, $a'_1,\ldots, a'_k \in \mathbb D$, $I_1,\ldots,I_l \in \mathbb I(\mathbb D)$, and $I'_1,\ldots,I'_l \in \mathbb I(\mathbb D')$ are such that $(a_j,a_j')\in\zeta_p$ for $1 \leq j \leq k$, and $(I_j,I'_j)\in\zeta_i$ for $1 \leq j \leq l$, then for every $FO+S$ formulas $\varphi(x_p^1,\ldots, x_p^k,y_i^1,\ldots, y_i^l)$ with free variables $x_p^1,\ldots x_p^k, y_i^1, \ldots y_i^l$, we have that
\[
\mathcal{F} \models \varphi(a_1, \ldots, a_k, I_1, \ldots, I_l) \text{ \ if and only if \ } \mathcal{F}' \models \varphi(a'_1, \ldots a'_k, I'_1, \ldots, I'_l).
\]
\end{thm}

\noindent Thus, to show that $FO + S \not \rightarrow r$ for a given $r\in\allr$, it is sufficient to find two concrete point-interval structures $\mathcal{F}$ and $\mathcal{F}'$ and a surjective $S$-truth preserving relation $\zeta$ between  $\mathcal{F}$ and $\mathcal{F}'$ which breaks $r$. For the readers' convenience, let us refer to surjective $S$-truth preserving relations as simply $S$-{\em relations}.

Although there are other constructions that could be used to show that relations are not definable in $FO + S$, e.g. elementary embeddings or Ehrenfeucht-Fra\"{i}ss\'{e} games, we have found $S$-relations sufficient for our purposes in this paper.

\subsection{Strategy}

The main objective of this paper is to establish all expressively different subsets of $\allr$ (and, then, of $\alli,\allm$ or $\allp$) over the mentioned classes of linear orders. To this end, for each $r\in\allr$ we compute all expressively different minimally $r$-complete and all maximally $r$-incomplete subsets of $\allr$, from which we can easily deduce all expressively different minimally $r$-complete and maximally $r$-incomplete subsets of $\alli,\allm$ and $\allp$; minimally $\allr$- (resp., $\alli-,\allm-,\allp-$) complete and maximally incomplete subsets are, then, deduced as a consequence of the above results. The set $\allr$ contains, as we have mentioned, 26 different relations. This means that there are $2^{26}$ potentially different extensions of first-order logic to be studied. Clearly, unless we design a precise strategy that allows us to reduce the number of results to be proved, the task becomes cumbersome.

\medskip

As a first simplification principle observe that, since we are working within first-order logic, all inverses of relations are explicitly definable, and hence we only need to assume as primitive a set which contains all relation up to inverses, which implies that point-interval relations can be omitted if we consider all interval-point ones. Accordingly, let $\alli^{+}$  be the set of interval-interval relations given in Tab.~\ref{ii:relations} together with $=_i$, $\allm^{+}$ be the set of interval-point relations given in Tab.~\ref{ip:relations}, and let $\allp^{+}=\{<,=_p\}$. Lastly let $\allr^{+}=\alli^{+} \bigcup \allm^{+} \bigcup \allp^{+}$.

\medskip

In order to further reduce the number of results to be presented, consider what follows. The {\em order dual} of a structure $\mathcal{F} = \langle\mathbb D,\mathbb I(\mathbb D)\rangle$ is the structure $\mathcal{F}^{\partial} = \langle\mathbb D^\partial,\mathbb I(\mathbb D^{\partial})\rangle$  based on the order dual $\mathbb D^{\partial}$ (obtained by reversing the order) of the underlying linear order $\mathbb D$. All classes considered in this paper are closed under taking order duals.

\begin{defi}
The \emph{reversible relations} are exactly the members of the set $\{\ip 0,\ip1,$ $\ip3,\ip4,$ $\ii 14,\ii 03\}$. The relations belonging to the complement  $\allr^{+} \setminus \{\ip 0,\ip1,$ $\ip3,\ip4,\ii 14,\ii 03\}$ are called \emph{symmetric}; if, in addition, $r=\ip2$ or $r=\ii04$, then $r$ is said {\em self-symmetric}. If $r=\ip 0$ (resp., $r=\ip 1, r=\ii14$), its {\em reverse} is $r=\ip 4$ (resp., $r=\ip 3,r=\ii 03$), and the other way around. Finally, the {\em symmetric} $S'$ of a subset $S\subseteq\allr^{+}$ is obtained by replacing every reversible relation in $S$ with its reverse. We shall use the notation $S \sim S'$ to indicate that sets $S$ and $S'$ are symmetric.
\end{defi}

\noindent This definition is motivated by the following easily verifiable facts. Let $r \in \allr^+$, $\mathcal{F}$ be a structure, and $x$ and $y$ be elements of $\mathcal{F}$ of the appropriate sorts for $r$; then:
\begin{enumerate}[label={(\emph{\roman*})}]
\item if $r$ is a reversible relation, with reverse $r'$, then $\mathcal{F} \models r(x,y)$ if and only if $\mathcal{F}^\partial \models r'(x,y)$;
\item if $r$ is self-symmetric, then $\mathcal{F} \models r(x,y)$ if and only if $\mathcal F^{\partial} \models r(x,y)$;
\item if $r$ is a symmetric, but not self-symmetric, relation, then $\mathcal F \models r(x,y)$ if and only if $\mathcal F^{\partial} \models r(y,x)$.
\end{enumerate}

\noindent The following crucial lemma capitalizes on these facts.

\begin{lem}\label{lem:symm}
Let $S, S' \subset\allr^{+}$ be such that $S \sim S'$. If $r$ is a symmetric relation, then $FO+S\rightarrow r$ if and only if  $FO+S'\rightarrow r$. Moreover, if $r$ is a reversible relation with reverse $r'$, then  $FO+S\rightarrow r$ if and only if  $FO+S'\rightarrow r'$.
\end{lem}

\proof Let $S, S' \subset\allr^{+}$ such that $S \sim S'$. For any $FO+S$ formula $\varphi$ that defines a given relation (and, therefore, with exactly two free variables), let $\varphi'$ denote the formula obtained from $\varphi$ by replacing every occurrence of a reversible relation with its reverse, and by swapping the arguments of every symmetric, but not self-symmetric, relation (occurrences of every self-symmetric relation are left unchanged). Induction on formulas then shows that $\mathcal F \models \varphi(x,y)$ (after substituting $x,y$ with elements of the appropriate sorts) if and only if $\mathbb \mathcal F^\partial \models \varphi'(x,y)$, for any structure $\mathcal F$. The base case of the induction is taken care of by the three observations preceding this lemma. Now, suppose that a $FO+S$ formula $\varphi(x,y)$ defines a symmetric relation $r$. We claim that $\varphi'$ also defines $r$. Let $\mathcal F$ be an arbitrary structure of signature $S\cup\{r\}$. Then $\mathcal F^{\partial} \models \varphi(x,y) \leftrightarrow r(x,y)$, and hence $\mathcal F \models \varphi'(x,y) \leftrightarrow r(y,x)$ if $r$ is not self-symmetric, and $\mathcal F \models \varphi'(x,y) \leftrightarrow r(x,y)$ otherwise. Next, suppose that the $FO+S$ formula $\varphi(x,y)$ defines a reversible relation $r$. We claim that $\varphi'$ defines its reverse $r'$. Let $\mathcal F$ be an arbitrary structure of signature $S\cup\{r\}$. Then $\mathcal F^{\partial} \models \varphi(x,y) \leftrightarrow r(x,y)$, and, hence, $\mathcal F \models \varphi'(x,y) \leftrightarrow r'(x,y)$.
\qed

\begin{figure}[t]
\scriptsize
\centering
\begin{code}{60mm}
\FUNCTION{Undef}{r\in\allr^{+}}
\BEGIN
\FORALL \ S\subset\allr^{+}\\
\BEGIN
S=Closure(S);\\
\IF ((r\notin S) \AND (S \ is \ maximal)) \THEN
list \ S
\END\\
\END \\
\end{code}
\begin{code}{65mm}
\FUNCTION{Closure}{set\ S, rules\ def\_rules}
\BEGIN
\WHILE  (S \ changes)\\
\BEGIN
\FORALL 1\le i\le size(def\_rules)\\
\BEGIN
\IF (def\_rules[i] \ applies) \THEN S=Apply(S,def\_rule[i])\\
\END\\
\END\\
\RETURN\ S
\END \\
\end{code}
\caption{Pseudo-code to identify maximally $r$-incomplete sets.}\label{fig:undef}
\end{figure}

In conclusion, we can limit our attention to 14 out of 26 relations by disregarding the inverses of relations in $\allr^{+}$, and we do not need to explicitly
analyze complete and incomplete sets for $\ip 3$, $\ip 4$, and $\ii 03$ as those correspond exactly to symmetric of complete and incomplete sets for $\ip 0$, $\ip 1$, and $\ii 14$, respectively. This means that only 11 relations are to be analyzed (which we can refer to as {\em explicit} relations).

\medskip

Even under the mentioned simplifications, there is a huge number of results to be presented and displayed. Let $r$ be anyone of the explicit relations.
In order to correctly identifying all minimally $r$-complete sets ($\mathsf{mcs}(r)$), we need to know all maximally $r$-incomplete sets ($\mathsf{MIS}(r)$) over the same class, and the other way around. To this end, we proceed in the following way:

\begin{enumerate}
\item fixed a class of linearly ordered sets and an explicit relation $r$, we first guess the $r$-complete subsets of $\allr^{+}$, obtaining
a first approximation of the definability rules for $r$;
\item then, we apply the algorithm in Fig.~\ref{fig:undef}, which uses the set of $r$-complete subsets of $\allr^+$ (the parameter $def\_rules$) to obtain a first approximation of the maximally $r$-incomplete sets;
\item after that, we prove that every $R_1,R_2,\ldots,R_k$ listed as a maximally $r$-incomplete set is actually $r$-incomplete, and, if not, we repeat from step 1, using the acquired knowledge to update the set of $r$-complete subsets of $\allr^+$;
\item at this point, the sets $S_1,S_2,\ldots,S_{k'}$ listed at step 1 are, actually, all minimally $r$-complete: for each $i$, $S_i$ is $r$-complete by definition, and if there was a $r$-complete set $S\subset S_i$, then for some $R_j$ listed as maximally $r$-incomplete set we could not prove its $r$-incompleteness. Therefore, $S_1,S_2,\ldots,S_{k'}$ are all minimally $r$-complete, and, as a consequence, $R_1,R_2,\ldots,R_k$ are all maximally $r$-incomplete.
\end{enumerate}

\noindent Once the above process is completed for every relation, we can then easily deduce all minimally $\allr^{+}$-complete and all maximally
$\allr^{+}$-incomplete sets, to complete the picture. A similar procedure works for $\alli^{+}, \allm^{+}$, and $\allp^{+}$.

\begin{table}[t]
\begin{center}
\begin{tabular}{|p{0.22\textwidth}|p{0.60\textwidth}|}
\hline
                  $\mathbb D,\mathbb D',\ldots$            & (generic) linearly ordered sets     \\
                  $x_p,y_p,\ldots$                         & first-order variables for points \\
                  $x_i,y_i,\ldots$                         & first-order variables for intervals \\
                  $x,y,\ldots$                             & first-order variables of any sort \\
                  {\em before},\ldots                       & relations in text are {\em emphasized}\\
                  $\mathcal F,\mathcal F',\ldots$          & (generic) concrete (point-)interval structures     \\
                  $S,S',\ldots$                            & (generic) subsets of $\allr$-relations  \\
                  $\zeta~(\zeta_p,\zeta_i)$                & surjective relation (for points, for intervals) \\
                  $Id_p (Id_i)$                            & `identity' relation over points (intervals) \\
                  $\mathrm C,\mathrm C'$                   & (generic) class of linearly ordered sets\\
                  $FO+S\rightarrow_{\mathrm C} r$          & $S$ defines $r$ (w.r.t. the class $C$)\\
                  $S\sim S'$                               & $S$ and $S'$ are symmetric\\
                  $a\in\mathbb D$                          & $a$ is a point of $D$, where $\mathbb D=(D,<)$\\
                  $\underline S$                           & in the text, a new proof case is \underline{underlined}\\
                  $r$                                      & generic relation \\
                  $\mathsf{mcs}$ ($\mathsf{mcs}(r)$)                         & minimally complete set (minimally $r$-complete set)\\
                  $\mathsf{MIS}$ ($\mathsf{MIS}(r)$)                         & maximally incomplete set (maximally $r$-incomplete set)\\
\hline
\end{tabular}
\end{center}\caption{Notational conventions used in this paper.}\label{tab:conv}
\end{table}

\medskip

\medskip

The most common notational conventions used in the paper are listed in Tab.~\ref{tab:conv}.

\section{Completeness Results in The Class \texorpdfstring{$\Lin$}{Lin}}\label{sec:lin}

In this section, we start analyzing the inter-definability of relations in $\allr^{+}$, and. In particular, we consider the case in which
we do not assume any particular property of the underlying linear order. It is convenient to begin by focusing
our attention to the sub-languages induced by $\allm^+$ and $\alli^+$; notice, in this respect, that while the semantic counterpart of the
sub-language $FO+\alli^{+}$ is essentially single-sorted (it is interpreted on interval structures), in the case of $FO+\allm^{+}$ (interpreted
on point-interval ones) both sorts are necessary. The results for $\allm^+$ and $\alli^+$
can also be found in~\cite{time2012}.

\medskip

Throughout our analysis we shall make extensive use of the following schema for the definability part: for every definability equation $r(x,y)\leftrightarrow \varphi(x,y)$, we denote by $\varphi(x,y)$ the right-hand part of the definition, indicating that $x,y$ are the only free variables in it; we then take a generic point-interval structure $\mathcal F= \langle \mathbb D,\mathbb I(\mathbb D)\rangle$, and show that $\mathcal F\models\varphi(x,y)$ (where $x,y$ have been instantiated with suitable constants of the right type) if and only if  $r(x,y)$ (again, after the due instantiation). We shall therefore omit the specification of these symbols and their meaning, as it remains the same in every proof. In order to make the text more readable, we shall present the results for each relation $r$ by means a table with at most four columns under the following headings:

\begin{enumerate}
\item {\em Proved}, which contains those $r$-complete sets for which we give an explicit proof in the corresponding lemma;
\item {\em Symmetric}, which contains, for each $r$-complete set listed in the {\em Proved} column, its symmetric one (if $r$ is symmetric);
\item {\em Implied}, which contains all $r$-complete sets that can be deduced from those in the first two columns plus the definability results presented earlier in the paper;
\item {\em Deduction Chain}, which is not empty if {\em Implied} is not empty, and it makes the chain of deductions explicit.
\end{enumerate}

\noindent When all (explicit) relations have been treated, we shall present the result of applying the algorithm in Fig.~\ref{fig:undef}, and we shall prove the undefinability results, completing the process and consequently proving the minimality of the complete sets.

\subsection{Definability in  $\alli^+$ and in $\allm^+$.}\label{subsec:allimlin}

We now study the minimal definability of $\alli^+$ relations, first, and, then, of $\allm^+$ relations.

\begin{table}[t]
\small
\begin{center}
\begin{tabular}{|p{0.1\textwidth}|p{0.17\textwidth}|p{0.17\textwidth}|p{0.17\textwidth}|p{0.2\textwidth}|}
\hline
Relation & Proved           & Symmetric          & Implied          & Deduction Chain\\
\hline
$\ii 34$ & $\{\ii14,\ii24\}$&  $\{\ii03,\ii24\}$ &                  &\\
         & $\{\ii14,\ii03\}$&  -                 &                  &\\
         & $\{\ii14,\ii44\}$&  $\{\ii03,\ii44\}$ &                  &\\
         & $\{\ii14,\ii04\}$&  $\{\ii03,\ii04\}$ &                  &\\
\hline
$\ii 14$ & $\ii34$          &  -                 & $\{\ii03,\ii24\}$&$\ii 34$\\
         &                  &                    & $\{\ii03,\ii44\}$&$\ii 34$\\
         &                  &                    & $\{\ii03,\ii04\}$&$\ii 34$\\
\hline
$\ii 24$ & $\ii34$          &  -                 &$\{\ii14,\ii03\}$ &$\ii 34$\\
         &                  &                    &$\{\ii14,\ii44\}$ &$\ii 34$\\
         &                  &                    &$\{\ii14,\ii04\}$ &$\ii 34$\\
         &                  &                    &$\{\ii03,\ii44\}$ &$\ii 34$\\
         &                  &                    &$\{\ii03,\ii04\}$ &$\ii 34$\\
\hline
$\ii 04$ & $\ii34$          &  -                 &$\{\ii14,\ii24\}$ &$\ii 34$ \\
         &                  &                    &$\{\ii14,\ii03\}$ &$\ii 34$ \\
         &                  &                    &$\{\ii14,\ii44\}$ &$\ii 34$ \\
         &                  &                    &$\{\ii03,\ii24\}$ &$\ii 34$ \\
         &                  &                    &$\{\ii03,\ii44\}$ &$\ii 34$ \\
\hline
$\ii 44$ & $\ii34$          &  -                 &$\{\ii14,\ii24\}$&$\ii 34$\\
         &                  &                    &$\{\ii14,\ii03\}$&$\ii 34$\\
         &                  &                    &$\{\ii14,\ii04\}$&$\ii 34$\\
         &                  &                    &$\{\ii03,\ii24\}$&$\ii 34$\\
         &                  &                    &$\{\ii03,\ii04\}$&$\ii 34$\\
\hline
$=_i$    & $\ii34$          &  -                 &   &  \\
         & $\ii 14$         & $\ii 03$           &   & \\
\hline
\end{tabular}
\caption{The spectrum of the $\mathsf{mcs}(r)$, for each $r\in\alli^+$. - Class: $\Lin$.}\label{tab:allilin} (Lemma~\ref{lem:tab:allilin}) 
\end{center}
\end{table}



\begin{lem}\label{lem:tab:allilin}
Tab.~\ref{tab:allilin} is correct.
\end{lem}

\begin{proof}
First, we prove the $\alli^+$-completeness of $\{\ii 34\}$, as well as the fact that every relation in $\alli^+$  is $=_i$-complete, and, then, we prove that every other subset is $\ii 34$-complete; completeness for the remaining relations is a mere consequence, as it can be seen in the table. As for the first step, we simply express every other interval-interval relation, as follows:

\medskip

\begin{small}
\[
\begin{array}{lll}
x_i\ii{4}{4}y_i & \leftrightarrow & \exists z_i(x_i\ii{3}{4} z_i\wedge z_i\ii{3}{4}y_i) \nonumber\\
x_i\ii{1}{4}y_i & \leftrightarrow & \exists z_i(x_i\ii{3}{4} z_i\wedge \forall k_i((z_i\ii{3}{4}k_i\leftrightarrow y_i\ii{3}{4}k_i)\wedge (k_i\ii{3}{4}x_i\leftrightarrow k_i\ii{3}{4}y_i)) \\
x_i\ii{0}{3}y_i & \leftrightarrow & \exists z_i(z_i\ii{3}{4} x_i\wedge \forall k_i((k_i\ii{3}{4}z_i\leftrightarrow k_i\ii{3}{4}y_i)\wedge
(x_i\ii{3}{4}k_i\leftrightarrow y_i\ii{3}{4}k_i))\\
x_i\ii{0}{4}y_i & \leftrightarrow & \exists z_i(x_i\ii{1}{4} z_i\wedge z_i\ii{0}{3}y_i) \\
x_i\ii{2}{4}y_i & \leftrightarrow & \exists z_i(z_i\ii{0}{3} x_i\wedge z_i\ii{1}{4}y_i) \\
x_i=_iy_i & \leftrightarrow & \forall z_i((x_i\ii 34 z_i\leftrightarrow y_i\ii34 z_i)\wedge(z_i\ii34 x_i\leftrightarrow z_i\ii34 y_i))\\
x_i=_iy_i & \leftrightarrow & \forall z_i((x_i\ii 14 z_i\leftrightarrow y_i\ii14 z_i)\wedge(z_i\ii14 x_i\leftrightarrow z_i\ii14 y_i))\\
\end{array}
\]
\end{small}

\medskip

\noindent All above equations but the last two are very simple, and do not require further explanation. Moreover, the $\alli^{+}$-completeness of $\{\ii{3}{4}\}$ is a known result (except for equality between intervals): it has been formally proved in~\cite{Allen85} assuming density and unboundedness of the structure, but a closer look shows that those additional hypothesis were not needed. As for the fact that $=_i$ can be expressed with $\{\ii 34\}$,
assume that $\mathcal{F} \models \varphi([a,b],[c,d])$. We wish to show that $a=c$ and $b=d$. Suppose, by way of contradiction, that $a\neq c$. If $a<c$, then the interval $[a,c]$ {\em meets} the interval $[c,d]$, but it does not {\em meet} $[a,b]$, contradicting the second conjunct of $\varphi$; if $c<a$, a similar argument shows that the second conjunct fails. Symmetrically, we can prove that $b\neq d$ leads to a contradiction with the first conjunct of $\varphi$. Also, it is obvious that if $[a,b]=[c,d]$, they {\em meet} and are {\em met by} exactly the same intervals. The case of $\{\ii 14\}$ is very similar. The remaining cases are more difficult, and require some non-trivial definitions. First, let us observe that having the weaker relation:

$$\ii 34\cup\ii 44$$

\noindent is enough to get $\ii 34$, and viceversa: in fact, $[a,b]$ {\em meets} $[b,c]$ if and only if $[a,b]\ii 34\cup\ii 44[b,c]$ and no other interval in between them has the same property. This is explicitly expressed by the formula $x_i34y_i \leftrightarrow x_i (\ii 34\cup\ii 44) y_i \wedge \neg \exists z_i (x_i (\ii 34\cup\ii 44) z_i \wedge z_i (\ii 34\cup\ii 44) y_i)$.  We use this observation in the rest of this proof, as in the remaining cases we are able to define precisely the relation $\ii 34\cup\ii 44$:

\medskip

\begin{small}
\[ x_i\ii 34\cup\ii 44 y_i \leftrightarrow \left\{ \begin{array}{ll}
                                  \neg(x_i\ii 24 y_i\vee y_i\ii 24 x_i\vee x_i\ii 14 y_i\vee y_i\ii 14 x_i)\wedge & \{\ii 14,\ii 24\}\\
                                  \exists z_i(x_i\ii 14 z_i\wedge\neg(y_i\ii 14 z_i))\wedge & \\
                                  \forall z_i(y_i\ii 14 z_i\rightarrow\neg (x_i\ii 24 z_i))\wedge & \\
                                  \forall z_it_i((z_i\ii 14 y_i\wedge x_i\ii 14 t_i)\rightarrow \neg(z_i\ii 24 t_i))\wedge & \\
                                  \forall z_i(x_i\ii 14 z_i\rightarrow \neg(y_i\ii 24 z_i))\wedge & \\
                                  \forall z_it_i((y_i\ii 14 z_i\wedge x_i\ii 14 t_i)\rightarrow \neg z_i\ii 24 t_i) & \\
                                  & \\
                                  \exists z_i(x_i\ii{1}{4}z_i\wedge y_i\ii{0}{3}z_i)\wedge & \{\ii 14,\ii 03\} \\
                                  \neg\exists z_i(z_i\ii{0}{3}x_i\wedge z_i\ii{1}{4}y_i) &  \\
                                  &\\
                                  \neg \exists z_i (\varphi_1(z_i,y_i) \wedge \varphi_2(z_i,x_i)) & \{\ii 14,\ii 44\} \\
   \end{array} \right. \]
\end{small}

\medskip

\noindent where:

\medskip

\begin{small}
\[\varphi_1(x_i,y_i)\leftrightarrow x_i\ii{1}{4} y_i \vee y_i \ii{1}{4}x_i \vee \forall z_i((z_i\ii14 x_i\leftrightarrow z_i\ii 14 y_i)\wedge(x_i\ii14 z_i\leftrightarrow y_i\ii 14 z_i))\]
\[\varphi_2(x_i,y_i) \leftrightarrow \forall z_i (x_i \ii{4}{4} z_i \leftrightarrow y_i \ii{4}{4} z_i)\wedge(\neg \exists z_i(x_i\ii{1}{4}z_i) \leftrightarrow \neg \exists z_i(y_i\ii{1}{4}z_i))\]
\end{small}

\medskip

\noindent As for the case of \underline{$\{\ii 14,\ii 24\}$}, assume $\mathcal{F} \models \varphi([a,b],[c,d])$; we want to prove that $b\le c$. It is
easy to see that the requirement excludes every other possibility. First, observe that $[a,b]$ and $[c,d]$ cannot {\em overlap} each other, nor can {\em
start} each other, thanks to the first line. The point $b$ cannot be the last one of the model as there must be an interval {\em started by} $x_i$ (second
line), and since such an interval cannot be {\em started by} $y_i$ (second line), $x_i$ and $y_i$ cannot be equal. If $y_i$ was {\em during} $x_i$, there would be an interval that is {\em started by} $y_i$ and {\em overlapped by} $x_i$, which is a contradiction (third line). If $x_i$ was {\em during} $y_i$, there would be an interval $z_i$ that {\em starts} $y_i$ and an interval $t_i$ {\em started by} $x_i$, and $z_i$ would {\em overlap} $t_i$, which is, again, a contradiction (fourth line). It is then easy to see that $x_i$ and $y_i$ cannot {\em finish} each other (third and fifth line); finally we would have a contradiction if $c\le a$ (fifth line). Thus, the only remaining possibility is the correct one. Conversely, if we assume that $b\le c$, it is straightforward to see that all requirements are respected. Let us now consider the case \underline{$\{\ii 14,\ii 44\}$}, which is slightly harder. Consider, first, the definition of $\varphi_1$: it is easy to see that if $\mathcal{F}$ is  concrete interval structure, then $\mathcal{F} \models \varphi_1([a,b],[c,d])$ if and only if $a=c$. Let us now analyze the definition of $\varphi_2$. Consider $[a,b],[c,d] \in \mathcal{F}$. If $b=d$ then it is clear that $\mathcal{F} \models \varphi_2([a,b],[c,d])$. Suppose that $\mathcal{F} \models \varphi_2([a,b],[c,d])$. We claim that $b=d$. Suppose, by way of contradiction that $ b \neq d$. As $\varphi_2 (x,y) \leftrightarrow \varphi_2(y,x)$ we may assume that $b<d$. If $d$ is the greatest point of the linear order $\mathbb D$ then the last conjunct of $\varphi_2$ does not hold. If there is a point $e$ which is greater than $d$, then $[a,b] \ii{4}{4} [d,e]$ and $\neg [c,d] \ii{4}{4} [d,e]$, which means that the first conjunct of $\varphi_2$ does not hold. So we have that  $\mathcal{F} \models \varphi_2([a,b],[c,d])$ if and only if $b=d$. Finally, we want to show that $\mathcal{F} \models [a,b]\ii34\cup\ii44[c,d]$ if and only if $\mathcal{F} \models \varphi([a,b],[c,d])$ where $\varphi$ denotes the right-hand part of the last equivalence considered for this set. Assume that $\mathcal{F} \models \varphi([a,b],[c,d])$. If $c<b$ then $z=[c,b]$ witnesses the failure of $\varphi([a,b],[c,d])$. So $b \leq c$ and hence $\mathcal{F} \models [a,b]\ii{3}{4}\vee\ii{4}{4}[c,d]$. Now assume that $\mathcal{F} \models [a,b]\ii{3}{4}\cup\ii{4}{4}[c,d]$, i.e., $b \leq c$. If $\mathcal{F} \models \varphi_1(z_i,[c,d])$, then $z_i=[c,e]$ with $c<e$ which implies $b<e$. So $\mathcal{F} \models \neg \varphi_2(z_i,[a,b])$. Therefore we have $\mathcal{F} \models \varphi([a,b],[c,d])$. For the case \underline{$\{\ii 14,\ii 03\}$}, assume that $\mathcal{F} \models \varphi([a,b],[c,d])$. We wish to show that $\mathcal{F} \models [a,b]\ii 34\cup\ii 44[c,d]$, i.e., that $b \leq c$. Suppose, by way of contradiction, that $c < b$. By assumption, there exists an interval $z_i = [e,f]$ such that $a = e < b < f$ and $e < c < d = f$. Then $a < c < b < d$, hence $[c,b] \ii{0}{3} [a,b]$ and $[c,b] \ii{1}{4} [c,d]$, contradicting $\mathcal{F} \models \neg\exists z_i((z_i\ii{0}{3}[a,b])\wedge (z_i\ii{1}{4}[c,d]))$. Conversely, suppose that $\mathcal{F} \models [a,b]\ii 34\cup\ii 44[c,d]$, i.e., $a < b \leq c < d$. Then the interval $z_i = [a,d]$ witnesses the first conjunct of the definition. Moreover, for any $z_i$, if $z_i\ii{0}{3} [a,b]$ then $z_i=[a',b]$. Then $\neg (z_i \ii{1}{4} [c,d])$, as $b \leq c$, proving that the second conjunct also holds.
Finally, as for the set \underline{$\{\ii 14,\ii 04\}$}, we can easily see that it is $\ii 34$-complete by means of an indirect definition, that is, by defining $\ii 24$:

\medskip

\begin{small}
\[
\begin{array}{lll}
x_i\ii 24 y_i &\leftrightarrow & \exists z_i(x_i\ii 14 z_i\wedge\neg(y_i\ii 04 z_i)\wedge\exists t_i(t_i\ii 04 z_i\wedge t_i\ii 14 y_i))\wedge\\
              &                & \exists t_i(t_i\ii 14 y_i\wedge\forall w_i(x_i\ii 04 w_i\rightarrow t_i\ii 04 w_i))
\end{array}
\]
\end{small}
\noindent whose correctness is immediate.\qedhere

\end{proof}

\medskip

We now focus our attention to $\allm^+$. Recall that models here are based on point-interval structures; we are therefore allowed to define
interval-interval relations whenever we need them.

\begin{table}[t]
\small
\begin{center}
\begin{tabular}{|p{0.1\textwidth}|p{0.20\textwidth}|p{0.20\textwidth}|p{0.20\textwidth}|}
\hline
Relation & Proved           & Implied          & Deduction Chain\\
\hline
$\ip0$   & $\{\ip 1,\ip 2\}$&$\{\ip 1,\ip 4\}$ & $\ip 0\cup\ip 2,\ip 3$\\
         & $\{\ip 1,\ip 3\}$&$\{\ip 2,\ip 3\}$ & $\ip 4$\\
         &                  &$\{\ip 2,\ip 4\}$ & $\ip 0\cup\ip 3,\ip 1$\\
\hline
$\ip 1$  & $\{\ip 0,\ip 3\}$&$\{\ip 2,\ip 3\}$ &$\ip 0,\ip 4$\\
\hline
$\ip 2$  & $\{\ip 0,\ip 4\}$&$\{\ip 0,\ip 3\}$ &$\ip 1$\\
         & $\{\ip 1,\ip 3\}$&$\{\ip 1,\ip 4\}$ &$\ip 0$\\
\hline
\end{tabular}
\caption{The spectrum of the $\mathsf{mcs}(r)$, for each $r\in\allm^+$. - Class: \texorpdfstring{$\Lin$}{Lin}. }\label{tab:allmlin} (Lemma~\ref{lem:tab:allmlin})
\end{center}
\end{table}

\begin{lem}\label{lem:tab:allmlin}
Tab.~\ref{tab:allmlin} is correct.
\end{lem}

\begin{proof}
Let us focus, first, on $\ip 0$, and consider the following definitions:

\medskip

\begin{small}
\[ x_i\ip 0 y_p \leftrightarrow \left\{ \begin{array}{ll}
         \exists z_iw_p(z_i\ip 1 y_p\wedge x_i\ip 1 w_p\wedge z_i\ip2 w_p)  &  \{\ip 1,\ip 2\}\\
         & \\
         \exists z_iw_p(x_i\ip 1 w_p\wedge z_i\ip 3 w_p\wedge z_i\ip1 y_p) &  \{\ip 1,\ip 3\}\\
         \end{array} \right. \]
\end{small}

\medskip

\noindent The two cases above, namely \underline{$\{\ip 1,\ip 2\}$} and \underline{$\{\ip 1,\ip 3\}$} are almost immediate to see. Now,
consider the case \underline{$\{\ip 1,\ip 4\}$}. We can prove that it defines $\ip 0$ by exclusion. As a matter of fact, we can see that:

\medskip

\begin{small}
\[
\begin{array}{lll}
x_i\ip 0\cup \ip{2}y_p & \leftrightarrow & \exists z_i(\forall w_p(z_i\ip 4 w_p\leftrightarrow x_i\ip 4 w_p)\wedge z_i\ip 1 z_p)\wedge \neg(x_i\ip 1 y_p)
\end{array}
\]
\end{small}

\medskip

\noindent In this way, we have that $\ip 3$ is definable by difference:

\begin{small}
\[
\begin{array}{lll}
x_i\ip 3 y_p & \leftrightarrow & \neg(x_i\ip 1 y_p\vee x_i\ip 4 y_p \vee x_i \ip 0\cup \ip 2 y_p)
\end{array}
\]
\end{small}

\noindent and, therefore, the set is $\ip 0$-complete by using $\{\ip 1,\ip 3\}$. Let us now consider the case of \underline{$\{\ip 2,\ip 3\}$}. To deal with it, we first observe that this set is $\ip 4$-complete, because $\{\ip 1,\ip 2\}$ defines $\ip 0$  and we can then use Lemma~\ref{lem:symm}. Now, we can directly define $\ip 0$:

\medskip

\begin{small}
\[
\begin{array}{llll}
x_i\ip 0 y_p & \leftrightarrow & \neg(x_i\ip 2 y_p)\wedge\neg(x_i\ip 3 y_p)\wedge\neg(x_i\ip 4 y_p)\wedge & \{\ip 2,\ip 3\}\\
             &                 & \exists z_i(\forall w_p(z_i\ip 4 w_p\leftrightarrow x_i\ip 4 w_p)\wedge \forall w_p(x_i\ip 2 w_p\rightarrow z_i\ip 2 w_p)\wedge\\
             &                 & \hspace{0.6cm} \exists w_p(z_i\ip 2 w_p\wedge\neg(x_i\ip 2 w_p))\wedge\neg(z_i\ip 2 y_p))\\
\end{array}
\]
\end{small}

\noindent As for the case \underline{$\{\ip 2,\ip 4\}$}, we reason in a similar way. By slightly modifying the above definition, we obtain a weaker relation:

\medskip

\begin{small}
\[
\begin{array}{llll}
x_i\ip 0\cup\ip 3 y_p & \leftrightarrow & \neg(x_i\ip 2 y_p)\wedge\neg(x_i\ip 4 y_p)\wedge & \{\ip 2,\ip 4\}\\
             &                 & \forall z_i((\forall w_p(z_i\ip 4 w_p\leftrightarrow x_i\ip 4 w_p)\wedge \forall w_p(x_i\ip 2 w_p\rightarrow z_i\ip 2 w_p)\wedge\\
             &                 & \hspace{0.8cm} \exists w_p(z_i\ip 2 w_p\wedge\neg(x_i\ip 2 w_p)))\rightarrow \neg(z_i\ip 2 y_p))\\
\end{array}
\]
\end{small}

\noindent Then, $\ip 1$ is defined by difference, and $\ip 0$-completeness becomes a consequence of the $\ip 0$-completeness of $\{\ip 1,\ip 2\}$, seen above.
Let us consider the $\ip 1$-completeness. First, as for \underline{$\{\ip 0,\ip 3\}$} we have that:

\medskip

\begin{small}
\[
\begin{array}{llll}
x_i\ip 1y_p & \leftrightarrow & \neg(x_i\ip 0 y_p)\wedge\neg(x_i\ip 3 y_p)\wedge & \{\ip 0,\ip 3\}\\
            &                 & \neg\exists z_i(z_i\ip 3 y_p\wedge \forall w_p(z_i\ip 0 w_p\leftrightarrow x_i\ip 0 w_p))
\end{array}
\]
\end{small}

\medskip

\noindent Then, for the case of \underline{$\{\ip 2,\ip 3\}$}, we already know that this set is $\ip 0$-complete, and therefore it must be also $\ip 1$-complete thanks to the above argument. We end this proof by analyzing the $\ip 2$-complete sets:

\medskip

\begin{small}
\[ x_i\ip 2 y_p \leftrightarrow \left\{ \begin{array}{ll}
         \neg(x_i\ip 0 y_p)\wedge\neg(x_i\ip 4 y_p)\wedge &  \{\ip 0,\ip 4\} \\
         \exists z_iz_p(\neg(x_i\ip 0 z_p)\wedge\neg(x_i\ip 4 z_p)\wedge\neg(z_i\ip 0 y_p)\wedge\neg(z_i\ip 4 y_p)\wedge z_i\ip 0 z_p)&\\
         \exists z_iz_p(\neg(x_i\ip 0 z_p)\wedge\neg(x_i\ip 4 z_p)\wedge\neg(z_i\ip 0 y_p)\wedge\neg(z_i\ip 4 y_p)\wedge z_i\ip 4 z_p) & \\
         & \\
         \exists z_iz_p(x_i\ip 1 z_p\wedge z_i\ip 1 z_p\wedge z_i\ip3 y_p)\wedge  &  \{\ip 1,\ip 3\}\\
         \exists z_iz_p(z_i\ip 1 y_p\wedge z_i\ip 3 z_p\wedge x_i\ip3 z_p)
         \end{array} \right.
\]
\end{small}

\medskip

\noindent All of the above are easy to prove. Also, the correctness of the remaining two sets is immediate: from \underline{$\{\ip 0,\ip 3\}$}
we define $\ip 1$ and, then, we use $\{\ip 1,\ip 3\}$, and from \underline{$\{\ip 1,\ip 4\}$} we define $\ip 0$, and, then, we use $\{\ip 0,\ip 4\}$.
\end{proof}

This concludes our preliminary analysis of the expressiveness of our languages when we limit ourselves to relation in $\alli^+$ and $\allm^+$. We shall use these results in the rest of this section, dealing with the expressiveness within $\allr^+$.

\subsection{Definability in $\allr^+$}

\begin{table}[t]
\small
\begin{center}
\begin{tabular}{|p{0.20\textwidth}|p{0.20\textwidth}|}
\hline
Proved & Symmetric \\
\hline
$\{<\}$ & -   \\
$\{\ip 1\}$ & $\{\ip 3\}$  \\
\hline
\end{tabular}
\quad
\begin{tabular}{|p{0.20\textwidth}|p{0.20\textwidth}|}
\hline
Proved & Symmetric \\
\hline
$\{\ip0, \ip2 \}$ & $\{\ip2, \ip4 \}$  \\
$\{\ip0, \ip3 \}$ & $\{\ip1, \ip4 \}$  \\
$\{\ip0, \ip4 \}$ & -  \\
$\{\ip1, \ip2 \}$ & $\{\ip2, \ip3 \}$ \\
$\{\ip1, \ip3 \}$ & - \\
$\{\ii14 \}$      & $\{\ii03 \}$ \\
$\{\ii34 \}$      & - \\
\hline
\end{tabular}
\caption{The spectrum of the $\mathsf{mcs}(=_p)$~(left) and of the $\mathsf{mcs}(=_i)$~(right). - Class: $\Lin$.}(Lemma~\ref{lem:tab:EqP} and Lemma~\ref{lem:tab:EqI})\label{tab:EqPEqI}
\end{center}
\end{table}

In the rest of this section, we assume that the set of relations is $\allr^+$; unlike the previous results, we shall treat the relations one by one.
We begin our study by considering those sets that define the equality between points; then we move to the equality between intervals, which is no more complicate than the previous one, although there are more ways to define $=_i$ than to define $=_p$.

\begin{lem}\label{lem:tab:EqP}
Tab.~\ref{tab:EqPEqI} (left) is correct.
\end{lem}

\begin{proof}
Consider the following definitions

\medskip

\begin{small}
\[ x_p=_p y_p \leftrightarrow \left\{ \begin{array}{ll}
         \neg(x_p<y_p)\wedge\neg(y_p<x_p)  &  \{<\}\\
         & \\
         \forall x_i(x_i\ip 1 x_p\leftrightarrow x_i\ip 1 y_p) & \{\ip 1\}\end{array} \right. \]
\end{small}

\medskip

\noindent The case of \underline{$\{<\}$} is trivial. As for the case of \underline{$\{\ip 1\}$}; suppose that $\mathcal{F} \models \varphi(a,b)$. Clearly, it implies that $a$ is the starting point of an interval if and only if $b$ is the starting point of that interval. Hence $a=b$. On the other hand it is obvious that if $a=b$, then $\mathcal{F} \models \varphi(a,b)$.
\end{proof}

\medskip

\begin{lem}\label{lem:tab:EqI}
Tab.~\ref{tab:EqPEqI} (right) is correct.
\end{lem}

\begin{proof}
Consider the following definition:

\medskip

\begin{small}
\[
\begin{array}{llll}
x_i=_i y_i &\leftrightarrow &
         \forall z_p(x_i\ip{r} z_p\leftrightarrow y_i\ip{r} z_p)\wedge\forall z_p(x_i\ip{r'} z_p\leftrightarrow y_i\ip{r'} z_p) & \{\ip r,\ip r'\}
\end{array}
\]
\end{small}

\medskip

\noindent  where $\{ \ip r,\ip r' \} = S$, for any $S$ in the left-hand part of the table with $S \subseteq \allm^+$. All such cases  are based on the same, simple observation: in order to constrain two intervals to be the same interval, it suffices to fix the two endpoints. This is to say that, for each side of the intervals, we can simply express the fact that they have the same sets of points in a given point-interval relation with it. So, for example, consider \underline{$\{\ip{0},\ip{2}\}$}. Assume that $\mathcal{F} \models \varphi([a,b],[c,d])$: we obtain $a=c$ from the first conjunct of $\varphi$, and $b=d$ from the second conjunct. On the other hand, it is immediate to see that if $[a,b]=[c,d]$ then $\mathcal{F} \models \varphi([a,b],[c,d])$. The basic idea is now clear: by means of $\ip 0$ we fix the left endpoints, and by means of $\ip 2$ we fix the right endpoint (in this particular case, $\ip 2$ serves the right side, but, for example, in the case of \underline{$\{\ip{2},\ip{4}\}$}, it would serve the left one). Notice that the only pairs missing from the list (and the list of symmetric sets)  are $\{\ip 0,\ip 1\}$ and its symmetric one, for which this idea does not apply (they are, in fact, $=_i$-incomplete). The remaining definitions are already included in Lemma~\ref{lem:tab:allilin}.
\end{proof}


The case of `less then' between point is the first non-trivial case, as there are already many possible different definitions.

\begin{lem}\label{lem:tab:Lt}
Tab.~\ref{tab:Lt} is correct.
\end{lem}

\begin{proof}
Consider, first, the following definitions:

\begin{small}
\[ x_p < y_p \leftrightarrow \left\{ \begin{array}{ll}
         \exists x_i(x_i\ip 1 x_p\wedge x_i\ip 3 y_p)  &  \{\ip 1,\ip 3\}\\
         & \\
         (\neg \exists z_i (z_i \ip{1} y_p) \wedge \exists z_i(z_i \ip{1} x_p)) \vee \exists z_i(z_i\ip{0}x_p \wedge z_i\ip{1}y_p) & \{\ip 0,\ip 1\} \\
         & \\
         (\exists x_i(x_i\ip 1 x_p)\wedge\neg\exists x_i(x_i\ip 1 y_p))\vee \exists x_i(x_i\ip 1 x_p\wedge x_i\ip 2 y_p) & \{\ip 1,\ip 2\} \\
         & \\
         \exists z_i(z_i\ip 3 y_p\wedge\forall k_i(k_i\ip 3 x_p\rightarrow\neg\forall z_p(z_i\ip 0 z_p\leftrightarrow k_i\ip 0 z_p))\wedge\neg(z_i\ip 3 x_p)) & \{\ip 0,\ip 3\}\\
         & \\
         (\exists x_i(x_i\ip 1 x_p)\wedge\neg\exists x_i(x_i\ip 1 y_p))\vee & \{\ip 1,\ii 03\} \\
         \exists x_iy_i(x_i\ip 1 x_p\wedge y_i\ip 1 y_p\wedge y_i\ii{0}{3} x_i)\\
\end{array} \right. \]
\end{small}

\begin{table}[t]
\small
\begin{center}
\begin{tabular}{|p{0.20\textwidth}|p{0.20\textwidth}|p{0.20\textwidth}|p{0.20\textwidth}|}
\hline
Proved & Symmetric & Implied & Deduction Chain\\
\hline
$\{\ip0, \ip1 \}$         &$\{\ip3, \ip4 \}$       &$\{\ip1, \ii14, \ii24 \}$& $=_i$ (Section~\ref{subsec:allimlin})\\
$\{\ip0, \ip3 \}$         &$\{\ip1, \ip4 \}$       &$\{\ip1, \ii14, \ii04 \}$& $=_i$ (Section~\ref{subsec:allimlin})\\
$\{\ip1, \ip2 \}$         &$\{\ip2, \ip3 \}$       &$\{\ip1, \ii14, \ii44 \}$& $=_i$ (Section~\ref{subsec:allimlin})\\
$\{\ip1, \ip3 \}$         &-                       &$\{\ip3, \ii24, \ii03 \}$& $=_i$ (Section~\ref{subsec:allimlin})\\
$\{\ip1, \ii03 \}$        &$\{\ip3, \ii14 \}$      &$\{\ip3, \ii03, \ii04 \}$& $=_i$ (Section~\ref{subsec:allimlin}) \\
$\{\ip1, \ii34 \}$        &$\{\ip3, \ii34 \}$      &$\{\ip3, \ii03, \ii44 \}$& $=_i$ (Section~\ref{subsec:allimlin})\\
$\{\ip1, \ii04, =_i \}$   &$\{\ip3, \ii04, =_i \}$ &                       &  \\
$\{\ip1, \ii24, =_i \}$   &$\{\ip3, \ii24, =_i \}$ &                       &  \\
$\{\ip1, \ii44, =_i \}$   &$\{\ip3, \ii44, =_i \}$ &                       &  \\
  \hline
\end{tabular}
\caption{The spectrum of the $\mathsf{mcs}(<)$. - Class: $\Lin$.} (Lemma~\ref{lem:tab:Lt})\label{tab:Lt}
\end{center}
\end{table}

\noindent The case of \underline{$\{\ip 1,\ip 3\}$} is straightforward. Consider the case of \underline{$\{\ip{0},\ip{1}\}$}. Assume that $\mathcal{F} \models \varphi(a,b)$. If the first disjunct of $\varphi$ is satisfied then $b$ is the largest point of the linear order $\mathbb D$ and $a$ is not the largest point of $\mathbb D$ which implies $a<b$ as required. If the second disjunct of $\varphi$ holds, witnessed by the interval $[c,d]$, then $a<c$ and $b=c$, again leading to $a<b$. On the other hand assume $a<b$. If $b$ is the largest point of $\mathbb D$ then $a$ is not, and hence the first disjunct of $\varphi$ holds. If $b$ is not the largest point of the the linear order $\mathbb D$ then we can pick $c \in \mathbb D$ with $b<c$ and the interval $[b,c]$ witnesses the second disjunct of $\varphi$. Now consider the case \underline{$\{\ip 0,\ip 3\}$}, and assume, first, that $\mathcal{F} \models \varphi(a,b)$. So, some interval $z_i$ {\em ends} at $b$; if, by contradiction, $a=b$, then $z_i$ also {\em ends} at $a$, contradicting the last conjunct, and if $b<a$, then the interval $k_i$ that {\em starts} at the beginning point of $z_i$ and ends at $a$ contradicts the second conjunct. If, on the other hand, we assume $a<b$, we can take the interval $z_i=[a,b]$ to satisfy $\varphi(a,b)$, and we make sure that $z_i$ does not {\em end} at $x_p$, nor any interval {\em ending} at $x_p$ may possibly {\em start} together with $z_i$. Next, consider the case \underline{$\{\ip 1,\ip 2\}$}. Assume that $\mathcal{F} \models \varphi(a,b)$. If the first disjunct of $\varphi$ holds then, $b$ is the largest point of $\mathbb D$ and $a$ is not which leads to $a<b$. If the second disjunct of $\varphi$ holds, witnessed by the interval $x_i=[c,d]$, then $a=c$ and $c<b<d$ which again leads to $a<b$. On the other hand, assume $a<b$, and let us prove that $\mathcal{F} \models \varphi(a,b)$. If $b$ is the greatest point of $\mathbb D$, then the first disjunct of $\varphi$ holds. If it is not, then, there is a $c \in \mathbb D$ such that $c>b>a$ and the second disjunct of $\varphi$ holds, witnessed by the interval $x_i=[a,c]$. Let us prove the $<$-completeness of \underline{$\{\ip 1,\ii{0}{3}\}$}. Assume that $\mathcal{F} \models \varphi(a,b)$. If the first disjunct of $\varphi$ holds then $b$ is the largest point of $\mathbb D$ and $a$ is not, which implies that $a<b$. Suppose that the second disjunct of $\varphi$ holds, witnessed by the intervals $x_i=[c,d]$ and $y_i=[e,f]$. Then $a=c$, $b=e$ and $c<e<f=d$, hence $a<b$. Now assume that $a<b$. If $b$ is the largest point of $\mathbb D$ then the first disjunct of $\varphi$ holds. Otherwise the second disjunct of $\varphi$ holds witnessed by the intervals $x_i=[a,c]$ and $y_i=[b,c]$ where $b<c$. Now, let us focus on the following group of definitions:

\medskip

\begin{small}
\[ x_p < y_p \leftrightarrow \left\{ \begin{array}{ll}
         (\exists x_i(x_i\ip 1 x_p)\wedge\neg\exists x_i(x_i\ip 1 y_p))\vee & \{\ip 1,\ii 34\} \\
         \exists x_iy_iz_i(x_i\ip 1 x_p\wedge y_i\ip 1 y_p\wedge z_i\ii34 y_i\wedge\forall t_i(t_i\ii34 x_i\leftrightarrow t_i\ii 34 z_i)) &\\
         & \\
         \exists x_i y_i(x_i\ip 1 x_p \wedge y_i \ip{1} y_p\wedge y_i \ii{0}{4} x_i) \vee & \{\ip1, \ii04, =_i \} \\
         (\exists x_i(x_i\ip 1 x_p)\wedge \neg \exists x_i(x_i \ip 1 y_p))\vee &  \\
          (\exists x_i y_i (x_i \ip{1} x_p \wedge y_i \ip{1} x_p \wedge x_i \neq y_i) \wedge & \\
          \neg \exists x_i y_i (x_i \ip{1} y_p \wedge y_i \ip{1} y_p \wedge x_i \neq y_i)) & \\
          & \\
          \exists x_i y_i(x_i\ip 1 x_p \wedge y_i \ip{1} y_p \wedge x_i \ii{2}{4} y_i) \vee & \{\ip1, \ii24, =_i \} \\
          (\exists x_i(x_i\ip 1 x_p)\wedge \neg \exists x_i(x_i \ip 1 y_p)) \vee &\\
          (\exists x_i y_i (x_i \ip{1} x_p \wedge y_i \ip{1} x_p \wedge x_i \neq y_i) \wedge & \\
          \neg \exists x_i y_i (x_i \ip{1} y_p \wedge y_i \ip{1} y_p \wedge x_i \neq y_i)) &  \\
          & \\
          \exists x_i (x_i\ip 1 x_p \wedge \forall y_i(y_i \ip{1} y_p \rightarrow \exists z_i(x_i \ii{4}{4}z_i \wedge \neg (y_i \ii{4}{4}z_i)))) \vee & \{\ip1, \ii44, =_i \} \\
					(\exists x_i(x_i\ip 1 x_p)\wedge \neg \exists x_i(x_i \ip 1 y_p)) \vee & \\
					(\exists x_i y_i (x_i \ip{1} x_p \wedge y_i \ip{1} x_p \wedge x_i \neq y_i) \wedge & \\                                                         \neg \exists x_i y_i (x_i \ip{1} y_p \wedge y_i \ip{1} y_p \wedge x_i \neq y_i))
\end{array} \right. \]
\end{small}

\medskip

\noindent The case of \underline{$\{\ip 1,\ii{3}{4}\}$} is very similar to the previous one. Indeed, if we assume that $\mathcal{F} \models \varphi(a,b)$,
and that the second conjunct of  $\varphi$ holds (if the first conjunct holds, then we reason as in the previous case), then, if $x_i=[c,d]$ and $y_i=[e,f]$, we have that $a=c$, $b=e$, and the interval $z_i=[a,b]$ must exist, so that it is {\em met by} exactly the same intervals that {\em meet} $x_i$. Now, assume that $a<b$, where $b$ is not the largest point of $\mathbb D$ (otherwise, the first disjunct of $\varphi$ holds). Then, the second disjunct of $\varphi$ holds witnessed by the intervals $x_i=[a,c]$, $y_i=[b,c]$, and $z_i=[a,b]$ where $b<c$. Now consider the case of \underline{$\{\ip 1,\ii{0}{4},=_i\}$}. Suppose that $\mathcal{F} \models \varphi(a,b)$. If the first disjunct of $\varphi$ holds, witnessed by the intervals $x_i=[c,d]$ and $y_i=[e,f]$, then $a=c$, $b=e$ and $c<e<f<d$ and we get $a<b$. If the second disjunct of $\varphi$ holds, then $b$ is the largest point of $\mathbb D$ and $a$ is not, so $a<b$. If the third disjunct of $\varphi$ holds, then there are at least two points in $\mathbb D$ which are greater than $a$ and there is at most one point in $\mathbb D$ greater than $b$ which again leads to $a<b$. Now, assume $a<b$; we want to show that $\mathcal{F} \models \varphi(a,b)$. If $b$ is the largest point of $\mathbb D$, then $a$ is not and therefore the second disjunct of $\varphi$ holds. If there is exactly one point in $\mathbb D$ which is greater than $b$, then there are at least two points in $\mathbb D$ greater than $a$ and therefore the third disjunct of $\varphi$ holds. Now we may assume that there are points $c,d \in \mathbb D$ with $a<b<c<d$. Then the first disjunct of $\varphi$ holds witnessed by the intervals $x_i=[a,d]$ and $y_i=[b,c]$.
The case of \underline{$\{\ip 1,\ii{2}{4},=_i\}$} is similar to the previous one. Lastly, we consider the case of \underline{$\{\ip 1,\ii{4}{4},=_i\}$}. Suppose that $\mathcal{F} \models \varphi(a,b)$ and that the first disjunct of $\varphi$ holds, witnessed by the interval $x_i=[a,c]$. Suppose, towards a contradiction, that $b<c$. Then the interval $y_i=[b,c]$ satisfies $y_i \ip{1} b$ and therefore there is an interval $z_i$ such that $[a,c] \ii{4}{4}z_i$ and $\neg([b,c] \ii{4}{4} z_i)$ which is impossible. Therefore $c \leq b$ and we obtain $a<b$. If the second or the third disjunct of $\varphi$ holds then we again obtain $a<b$ precisely as in the previous case. Now assume $a<b$; we want to show that $\mathcal{F} \models \varphi(a,b)$. If $b$ is the greatest point of $\mathbb D$ then the second disjunct of $\varphi$ holds and if there is exactly one point in $\mathbb D$ which is greater than $b$ then the third disjunct of $\varphi$ holds. So we may assume there are two points $c,d \in \mathbb D$ such that $b<c<d$. We claim that the first disjunct of $\varphi$ holds, witnessed by the interval $x_i=[a,b]$.
Let $y_i=[b,e]$. We want to find an interval $z_i$ such that $[a,b] \ii{4}{4} z_i$ and $\neg [b,e] \ii{4}{4} z_i$. Let $f$ be the minimum of $c$ and $e$. Then the interval $z_i=[f,d]$ has the desired property, from which we see that the first disjunct of $\varphi$ holds.
\end{proof}

\medskip

We are now moving to mixed relations, starting with $\ip 0$. Recall that $\ip 0$ and $\ip 1$ are reversible relations: symmetric sets of complete ones are
complete for their (respective) reverse, and, thus, they do not appear in their tables.

\begin{lem}\label{lem:tab:Ip0}
Tab.~\ref{tab:Ip0} is correct.
\end{lem}

\begin{table}[t]
\small
\begin{center}
\begin{tabular}{|p{0.20\textwidth}|p{0.20\textwidth}|p{0.20\textwidth}|}
\hline
Proved & Implied & Deduction Chain\\
\hline
$\{\ip1, < \}$            & $\{\ip1, \ip2 \}$       & $<$\\
$\{\ip2, \ip3 \}$         & $\{\ip1, \ip3 \}$       & $<$\\
$\{\ip2, \ip4 \}$         & $\{\ip1, \ip4 \}$       & $<$\\
$\{\ip2, \ii14, < \}$     &$\{\ip1, \ii14, \ii24 \}$& $<$\\
$\{\ip2, \ii03, < \}$     &$\{\ip1, \ii14, \ii04 \}$& $<$\\
$\{\ip3, \ii14 \}$        &$\{\ip1, \ii14, \ii44 \}$& $<$\\
$\{\ip4, \ii14, < \}$     &$\{\ip1, \ii24, =_i \}$  & $<$\\
                          &$\{\ip1, \ii04, =_i \}$  & $<$\\
                          &$\{\ip1, \ii44, =_i \}$  & $<$\\
                          &$\{\ip1, \ii03 \}$       & $<$\\
                          &$\{\ip1, \ii34 \}$       & $<$\\
                          &$\{\ip2, \ii34, < \}$      & $\ii 03$ (Section~\ref{subsec:allimlin})\\
                          &$\{\ip3, \ii34 \}$         & $\ii 14$ (Section~\ref{subsec:allimlin})\\
                          &$\{\ip3, \ii24, \ii03 \}$  & $\ii 34$ (Section~\ref{subsec:allimlin})                    \\
                          &$\{\ip3, \ii03, \ii04 \}$  & $\ii 34$ (Section~\ref{subsec:allimlin})\\
                          &$\{\ip3, \ii03, \ii44 \}$  & $\ii 34$ (Section~\ref{subsec:allimlin})\\
                          &$\{\ip4, \ii34, < \}$         & $\ii 14$  (Section~\ref{subsec:allimlin})\\
                          &$\{\ip4, \ii24, \ii03, < \}$  & $\ii 34$  (Section~\ref{subsec:allimlin})\\
                          &$\{\ip4, \ii03, \ii04, < \}$  & $\ii 34$  (Section~\ref{subsec:allimlin})\\
                          &$\{\ip4, \ii03, \ii44, < \}$  & $\ii 34$  (Section~\ref{subsec:allimlin})\\
\hline
\end{tabular}
\caption{The spectrum of the $\mathsf{mcs}(\ip 0)$. - Class: $\Lin$.}(Lemma~\ref{lem:tab:Ip0})\label{tab:Ip0}
\end{center}
\end{table}

\begin{proof}
Consider the following definitions:
\medskip

\begin{small}
\[ x_i\ip 0 y_p \leftrightarrow \left\{ \begin{array}{ll}
         \exists z_p(x_i\ip 1 z_p\wedge y_p<z_p)  &  \{\ip 1,<\}\\
         & \\
         \exists z_p(\neg(x_i\ip2 z_p)\wedge\neg(x_i\ip 3 z_p)\wedge & \{\ip 2,\ip 3\} \\
\hspace{0.7cm}\forall z_i\forall k_p((z_i\ip 3 z_p\wedge x_i\ip 3 k_p)\rightarrow \neg (z_i\ip 2 k_p))\wedge y_p<z_p) & \\
\end{array} \right. \]
\end{small}

\medskip

\noindent The case of \underline{$\{\ip 1,<\}$} is straightforward, and needs no explanation. Consider, now, the case \underline{$\{\ip 2,\ip 3\}$}. The fact that this set is $<$-complete is proved above. So, assume first $\mathcal{F} \models \varphi([a,b],c)$; there must be some $d$ which, thanks to the first two conjuncts, can only be placed in such a way that $d\le a$ or $b<d$, and, thanks to the third conjunct, the possibility $b<d$ is eliminated: in fact, if we had $b < d$,  we could take $z_i=[a,d]$ and $k_p = b$ to contradict the third conjunct. So, since $d\le a$ and $c<d$, it must be that $c<a$ as we wanted. Conversely, suppose that $c<a$: we take $z_p=a$, which is not {\em during} $x_i$ nor does it {\em end} $x_i$, and it is such that no interval {\em ending} at $z_p$ may possibly {\em contain} $b$. Moreover, $y_p=c<a$, and so $\varphi([a,b],c)$. Now, focus on the following definitions:

\medskip

\begin{small}
\[ x_i\ip 0 y_p \leftrightarrow \left\{ \begin{array}{ll}

         & \\
         \neg(x_i\ip 2 y_p)\wedge\neg(x_i\ip 4 y_p)\wedge & \{\ip 2,\ip 4\}\\
\exists z_i(\forall k_p(x_i\ip 4 k_p\leftrightarrow z_i\ip 4 k_p)\wedge & \\
\exists k_p(z_i\ip 2 k_p\wedge\neg(x_i\ip 2 k_p))\wedge \neg(z_i\ip 2 y_p))\wedge & \\
\neg\forall z_i(\exists k_p(z_i\ip 4 k_p\wedge\neg (x_i\ip 4 k_p))\rightarrow z_i\ip 4 y_p)\\
&\\
\exists z_pk_p(y_p<z_p\wedge z_p<k_p\wedge\neg(x_i\ip 2 y_p)\wedge & \{\ip 2,\ii14, <\} \\
             \hspace{1cm}\neg(x_i\ip 2 z_p)\wedge\neg(x_i\ip 2 k_p)\wedge & \\
               \hspace{1cm}\forall z_i(z_i\ip 2 z_p\rightarrow\neg(x_i\ii 14 z_i)))  \\

\end{array} \right. \]
\end{small}

\medskip

\noindent Let us focus on the the case of \underline{$\{\ip 2,\ii{1}{4},<\}$}. Suppose $\mathcal{F} \models \varphi([a,b],c)$. So, there are two points $d,e$ such that $c<d<e$. Now  $d=z_p$ cannot be {\em after} $a$, for if it was we could find an interval $z_i=[a,e]$, where $e>b$ (since $e$ and $d$ cannot be {\em during} $x_i$) and such that $x_i$ {\em starts} $z_i$, which contradicts the last conjunct. Then, $d\le a$, which implies $c<a$ as we wanted. Suppose, on the other hand, that $c<a$: to satisfy the requirements, it suffices to take $z_p=a$ and $k_p=b$. The correctness of the case \underline{$\{\ip 2,\ip 4\}$} is based on the fact that we can, first, eliminate the possibility that $y_p$ is {\em during} or {\em after} $x_i$; then, we eliminate the possibility that $y_p$ starts $x_i$ by stating that there must be an interval {\em finished by} $x_i$ that does not have $y_p$ {\em during} it; finally, we eliminate the possibility that $y_p$ {\em ends} $x_i$ by stipulating the existence of an interval (the one that {\em starts} at $y_p$) which {\em ends} before the right endpoint of $x_i$ that does not have $y_p$ {\em after} it. Finally, consider the following definitions:

\begin{small}
\[ x_i\ip 0 y_p \leftrightarrow \left\{ \begin{array}{ll}
\exists z_pz_i(x_i\ii03z_i\wedge z_i\ip 2z_p\wedge\neg(x_i\ip 2 z_p)\wedge y_p<z_p)& \{\ip 2,\ii03, <\} \\
               & \\
\exists z_p(\neg(x_i\ip 3 z_p)\wedge\neg \exists z_i((z_i\ii 14 x_i\vee x_i\ii 14 z_i)\wedge z_i\ip 3 z_p)\wedge & \{\ip 3,\ii 14\}\\
             \hspace{0.7cm}y_p<z_p) \\
&\\
\exists z_p(\neg\forall k_p(z_p<k_p\leftrightarrow x_i\ip 4 k_p)\wedge & \{\ip 4,\ii 14,<\} \\
\hspace{0.7cm}\neg\exists z_i((z_i\ii 14 x_i\vee x_i\ii 14 z_i)\wedge\forall k_p(z_p<k_p\leftrightarrow z_i\ip 4 k_p))\wedge &  \\
\hspace{0.7cm}y_p<z_p)\\
\end{array} \right. \]
\end{small}

\medskip

\begin{table}[t]
\small
\begin{center}
\begin{tabular}{|p{0.20\textwidth}|p{0.20\textwidth}|p{0.20\textwidth}|}
\hline
Proved & Implied & Deduction Chain\\
\hline
$\{\ip 0,<\}$            &$\{\ip0, \ip3 \}$        & $<$    \\
                         &$\{\ip2, \ip3 \}$        & $<,\ip 0$   \\
                         &$\{\ip2, \ip4, < \}$     & $\ip 0$\\
                         &$\{\ip2, \ii14, <\}$     & $\ip 0$\\
                         &$\{\ip2, \ii03, <\}$     & $\ip 0$\\
                         &$\{\ip2, \ii34, <\}$     & $\ip 0$\\
                         &$\{\ip3, \ii14 \}$       & $<,\ip 0$ \\
                         &$\{\ip3, \ii24, \ii03 \}$& $<,\ip 0$\\
                         &$\{\ip3, \ii03, \ii04 \}$& $<,\ip 0$\\
                         &$\{\ip3, \ii03, \ii44 \}$& $<,\ip 0$\\
                         &$\{\ip3, \ii34 \}$       & $<,\ip 0$\\ &$\{\ip4, \ii14, < \}$    &     $\ip 0$\\
                         &$\{\ip4, \ii24, \ii03, < \}$ & $\ip 0$\\
                         &$\{\ip4, \ii03, \ii04, < \}$ & $\ip 0$\\
                         &$\{\ip4, \ii03, \ii44, < \}$ & $\ip 0$\\
                         &$\{\ip4, \ii34, < \}$ &        $\ip 0$\\
\hline
\end{tabular}
\caption{The spectrum of the $\mathsf{mcs}(\ip 1)$. - Class: $\Lin$.}(Lemma~\ref{lem:tab:Ip1})\label{tab:Ip1}
\end{center}
\end{table}

\medskip

\noindent For the sake of the case \underline{$\{\ip 2,\ii{0}{3},<\}$}, suppose that $\mathcal{F} \models \varphi([a,b],c)$. So, there is a point $z_p$ such that it is not {\em during} $x_i$, and since it must be {\em during} an interval $z_i=[e,b]$, it can only be $d\le a$. Since $c<d$ we have $c<a$. Suppose, on the other hand, that $c<a$: to satisfy the requirements, it suffices to take $z_i=[c,b]$ and $z_p=a$. As for the set \underline{$\{\ip 3,\ii{1}{4}\}$}, first, recall its $<$-completeness; then, we state the existence of a point for which we can eliminate the possibility that $z_p$ {\em ends} $x_i$ by means of the first conjunct, and, after that, we can eliminate the possibility that $z_p$ is {\em during} or {\em after} $x_i$: if that were the case, there would be an interval $z_i$ {\em starting} or {\em
started by} $x_i$ such that $z_p$ is its right endpoint. The only possibility left is therefore that $z_p$ is less than or equal to the beginning point of $x_i$, and therefore the last conjunct guarantees that $x_i \ip 0 y_p$. Finally, the case \underline{$\{\ip 4,\ii 14,<\}$} is identical to the previous one, with the only difference that $x_i\ip 3 y_p$ can be expressed by asserting that $y_p<k_p$ and $x_i\ip 4 k_p$ are equivalent. As for the implied definitions, notice that in some cases we exploit the $\alli^{+}$-completeness results from Section~\ref{subsec:allimlin}.
\end{proof}

\begin{lem}\label{lem:tab:Ip1}
Tab.~\ref{tab:Ip1} is correct.
\end{lem}

\begin{proof}
Only one new definition is needed:

\medskip

\begin{small}
\[
\begin{array}{llll}
x_i\ip 1 y_p &\leftrightarrow& \forall z_p(x_i\ip 0 z_p\leftrightarrow z_p<y_p) & \{\ip 0,<\}
\end{array}
 \]
\end{small}

\noindent This definition is quite straightforward: $[a,b]\ip 1 a$ if, and only if, the points $c$ less than $a$ are exactly the same points $c$ such that $[a,b]\ip 0 c$.
\end{proof}
\medskip

\begin{table}[t]
\small
\begin{center}
\begin{tabular}{|p{0.20\textwidth}|p{0.20\textwidth}|p{0.20\textwidth}|p{0.20\textwidth}|}
\hline
Proved & Symmetric & Implied & Deduction Chain\\
\hline
$\{\ip0, \ip4 \}$                &-                 &$\{\ip0, \ip3 \}$             & $\ip 1$ \\
$\{\ip1, \ip3 \}$                &-                 &$\{\ip0, \ii14, \ii24, < \}$ & $\ip 4$ \\
$\{\ip1, \ii03\}$                &$\{\ip3, \ii14\}$ &$\{\ip0, \ii14, \ii04, < \}$ & $\ip 4$ \\
                                 &                  &$\{\ip0, \ii14, \ii44, < \}$ & $\ip 4$ \\
                                 &                  &$\{\ip0, \ii03, < \}$        & $\ip 1$ \\
                                 &                  &$\{\ip0, \ii34, < \}$        & $\ii 03$ (Section~\ref{subsec:allimlin})\\
                                 &                  &$\{\ip1, \ip4 \}$            & $\ip 3$\\
                                 &                  &$\{\ip1, \ii14, \ii24 \}$    & $\ip 3$ \\
                                 &                  &$\{\ip1, \ii14, \ii04 \}$    & $\ip 3$ \\
                                 &                  &$\{\ip1, \ii14, \ii44 \}$    & $\ip 3$ \\
                                 &                  &$\{\ip1, \ii34 \}$           & $\ii 03$ (Section~\ref{subsec:allimlin})\\
                                 &                  &$\{\ip3, \ii24, \ii03 \}$    & $\ip 1$\\
                                 &                  &$\{\ip3, \ii03, \ii04 \}$    & $\ip 1$\\
                                 &                  &$\{\ip3, \ii03, \ii44 \}$    & $\ip 1$\\
                                 &                  &$\{\ip3, \ii34 \}$           & $\ii 14$ (Section~\ref{subsec:allimlin})       \\
                                 &                  &$\{\ip4, \ii14, < \}$        & $\ip 3$      \\
                                 &                  &$\{\ip4, \ii24, \ii03, < \}$ & $\ip 1$\\
                                 &                  &$\{\ip4, \ii03, \ii04, < \}$ & $\ip 1$\\
                                 &                  &$\{\ip4, \ii03, \ii44, < \}$ & $\ip 1$ \\
\hline
\end{tabular}
\caption{The spectrum of the $\mathsf{mcs}(\ip 2)$. - Class: $\Lin$.}(Lemma~\ref{lem:tab:Ip2})\label{tab:Ip2}
\end{center}
\end{table}

We complete this part by analyzing the definability for $\ip 2$. Notice that here we also take advantage of the results
of Section~\ref{subsec:allimlin}.

\begin{lem}\label{lem:tab:Ip2}
Tab.~\ref{tab:Ip2} is correct.
\end{lem}

\begin{proof}
In this case, three new definitions are needed:


%
%

\medskip
\begin{small}
\[ x_i\ip 2 y_p \leftrightarrow \left\{ \begin{array}{ll}
         \exists z_ik_iz_pk_p(z_i\ip 1 z_p\wedge k_i\ip 3 k_p\wedge z_i\ip 3 y_p\wedge k_i\ip 1 y_p\wedge x_i\ip 1 z_p\wedge x_i\ip 3 k_p) & \{\ip 1,\ip 3\} \\
         & \\
         \exists z_ik_i(\forall z_p(z_i\ip 0 z_p\leftrightarrow x_i\ip 0 z_p)\wedge\forall k_p(k_i\ip 4 k_p\leftrightarrow x_i\ip 4 k_p)\wedge& \{\ip 0,\ip 4\}\\
\hspace{1cm}\neg\exists t_i(\forall t_p(t_i\ip0 t_p\leftrightarrow k_i\ip 0 t_p)\wedge\forall t_p(t_i\ip4 t_p\leftrightarrow z_i\ip 4 t_p))\wedge & \\
                      \hspace{1cm} \neg(z_i\ip 4 y_p)\wedge\neg(k_i\ip 0 y_p)) & \\
                      & \\
         \exists z_i(z_i \ii 03 x_i \wedge z_i \ip 1 y_p) & \{\ip 1,\ii 03\} \\
   \end{array} \right. \]
\end{small}
\medskip

\noindent For the $\ip 2$-completeness of \underline{$\{\ip 1,\ip 3\}$}, we begin by assuming that $\mathcal{F} \models \varphi([a,b],c)$. So, there must be two intervals $z_i,k_i$, the former {\em starting} at $a$, and the latter {\em ending} at $b$; the point $c$ must, at the same time, {\em end} $z_i$ (meaning that $a<c$) and {\em start} $k_i$ (implying that $c<b$), and we deduce $a<c<b$. If we assume $a<c<b$, it is enough to take $z_i=[a,c]$ and $k_i=[c,b]$ to satisfy the conditions. We move now to the $\ip 2$-completeness of \underline{$\{\ip 0,\ip 4\}$}, and let us assume, again, that $\mathcal{F} \models \varphi([a,b],c)$. The interval $z_i$ must {\em start} at $a$, and the interval $k_i$ must {\em end} at $b$. Moreover, the ending point of $k_i$ cannot be {\em after} the beginning point of $z_i$ (third conjunct). Now, in this situation, $y_p$ cannot be $a$ or {\em before} it (because of the fifth conjunct), and it cannot be $b$ or {\em after} it (because of the fourth conjunct), and the only possibility left out is $a<c<b$. If, conversely, we assume $a<c<b$, it is enough to take $z_i=[a,c]$ and $k_i=[c,b]$ to satisfy the entire set of conditions. Finally, the $\ip  2$-completeness of \underline{$\{\ip 1,\ii 03\}$} is very easy. If $\mathcal{F} \models \varphi([a,b],c)$, then there exists an interval $z_i$ that {\em finishes} $[a,b]$, and for its {\em starting} point, that is, $y_p=c$, it must hold that $a<c<b$. On the other hand, if $a<c<b$, we just take $z_i=[c,b]$ to satisfy $\varphi$.
\end{proof}

\medskip

We now focus on interval-interval relations other than equality, starting with $\ii 34$, which, being $\alli^+$-complete on its own, plays a very special
role. Most of the work has already been taken care of in Section~\ref{subsec:allimlin}. As a matter of fact, we already know
some subsets of $\alli^+$ that are $\ii 34$-complete, and we certainly have to define those subsets of $\allr^+\setminus\alli^+$
which are. It remains to establish whether other subsets of $\alli^+$, other than those seen in Section~\ref{subsec:allimlin}, become
$\ii 34$-complete within the language of $FO+\allr^+$: the fact that this is not the case will be a consequence of the results
shown in Section~\ref{sec:lin-incomp}.

\begin{table}[t]
\small
\begin{center}
\begin{tabular}{|p{0.20\textwidth}|p{0.20\textwidth}|p{0.20\textwidth}|p{0.20\textwidth}|}
\hline
Proved & Symmetric & Implied & Deduction Chain\\
\hline
$\{\ip 1,\ip 3 \}$               & -                 & $\{\ii 24,\ii 14 \}$& Section~\ref{subsec:allimlin}\\
$\{\ip 2,\ii 14 \}$              & $\{\ip2, \ii03 \}$& $\{\ii 14,\ii 44 \}$& Section~\ref{subsec:allimlin}\\
$\{\ip 4,\ii 14 \}$              & $\{\ip0, \ii03 \}$& $\{\ii 14,\ii 03 \}$& Section~\ref{subsec:allimlin}\\
$\{\ip0, \ip2 \}$                & $\{\ip2, \ip4 \} $& $\{\ii 14, \ii04 \}$& Section~\ref{subsec:allimlin}\\
                                 &                   & $\{\ii24, \ii03 \}$ & Section~\ref{subsec:allimlin}\\
                                 &                   & $\{\ii03, \ii04 \}$ & Section~\ref{subsec:allimlin}\\
                                 &                   & $\{\ii03, \ii44 \}$ & Section~\ref{subsec:allimlin}\\
                                 &                   & $\{\ip0, \ip3 \}$   & $\ip 1$ \\
                                 &                   & $\{\ip0, \ip4 \}$   & $\ip 2$ \\
                                 &                   & $\{\ip1, \ip2 \}$   & $\ip 0$ \\
                                 &                   & $\{\ip1, \ip4 \}$   & $\ip 3$ \\
                                 &                   & $\{\ip1, \ii03 \}$  & $\ip 2$ \\
                                 &                   & $\{\ip2, \ip3 \}$   & $\ip 4$ \\
                                 &                   & $\{\ip3, \ii14 \}$  & $\ip 2$ \\
\hline
\end{tabular}
\caption{The spectrum of the $\mathsf{mcs}(\ii 34)$. - Class: $\Lin$.}(Lemma~\ref{lem:tab:Ii34})\label{tab:Ii34}
\end{center}
\end{table}

\begin{lem}\label{lem:tab:Ii34}
Tab.~\ref{tab:Ii34} is correct.
\end{lem}

\begin{proof}
The fact that the set \underline{$\{\ip 1,\ip 3\}$} is $\ii 34$-complete depends on the following easy definition:

\medskip

\begin{small}
\[
\begin{array}{llll}
x_i\ii 34 y_i &\leftrightarrow& \exists z_p(x_i\ip 3 z_p\wedge y_i\ip 1 z_p) & \{\ip 1,\ip 3\}
\end{array}
 \]
\end{small}

\medskip
\noindent The second two sets in the leftmost column can be proven $\ii34$-complete by means of defining, as in Section~\ref{subsec:allimlin}, the weaker relation $\ii 34\cup\ii 44$:

\medskip

\begin{small}
\[ x_i\ii 34\cup\ii 44 y_i \leftrightarrow \left\{ \begin{array}{ll}
                                   \exists z_i(x_i\ii14 z_i\wedge\forall k_p(z_i\ip 4 k_p\leftrightarrow y_i\ip 4 k_p))\wedge & \{\ip 4,\ii 14\}\\
                                   \neg\exists z_i(z_i\ii14 y_i\wedge\forall k_p(z_i\ip 4 k_p\leftrightarrow x_i\ip 4 k_p)) &\\
                                   &\\
                                   \exists z_pz_i(x_i\ii14 z_i\wedge z_i\ip 2 z_p\wedge \neg(x_i\ip 2 z_p)\wedge & \{\ip 2,\ii 14\}\\
                                   \hspace{0.8cm}\forall k_i(x_i\ii 14 k_i\rightarrow k_i\ip 2 z_p)\wedge & \\
                                   \hspace{0.8cm}\forall k_i(y_i\ii 14 k_i\rightarrow \neg(k_i\ip 2 z_p))\wedge & \\
                                   \hspace{0.8cm}\neg(y_i\ip 2 z_p)), & \\
   \end{array} \right. \]
\end{small}

\medskip

\begin{table}[t]
\small
\begin{center}
\begin{tabular}{|p{0.20\textwidth}|p{0.20\textwidth}|p{0.20\textwidth}|}
\hline
Proved & Implied & Deduction Chain\\
\hline
$\{\ip 0,\ip2 \}$                &$\{\ip0, \ip3 \}$& $\ip 2$  \\
$\{\ip 0,\ip4 \}$                &$\{\ip0, \ii03, <\}$& $\ip 2$  \\
                                 &$\{\ip1, \ip2 \}$& $\ip 0$  \\
                                 &$\{\ip1, \ip3 \}$& $\ip 0,\ip4$  \\
                                 &$\{\ip1, \ip4 \}$& $\ip 0$  \\
                                 &$\{\ip1, \ii03\}$& $\ii 34$  \\
                                 &$\{\ip2, \ip3 \}$& $\ip 0$  \\
                                 &$\{\ip2, \ip4 \}$& $\ip 0$  \\
                                 &$\{\ii34 \}$      &Section~\ref{subsec:allimlin}\\
                                 &$\{\ip2, \ii03\}$ &$\ii34$  \\
                                 &$\{\ii24, \ii03\}$&Section~\ref{subsec:allimlin}\\
                                 &$\{\ii03, \ii04\}$&Section~\ref{subsec:allimlin}\\
                                 &$\{\ii03, \ii44\}$&Section~\ref{subsec:allimlin}\\
\hline
\end{tabular}
\caption{The spectrum of the $\mathsf{mcs}(\ii 14)$. - Class: $\Lin$.}(Lemma~\ref{lem:tab:Ii14}.)\label{tab:Ii14}
\end{center}
\end{table}

\noindent As for \underline{$\{\ip 4,\ii 14\}$}, assume, first, that $\mathcal{F} \models \varphi([a,b],[c,d])$. We wish to show that $\mathcal{F} \models [a,b]\ii 34\cup\ii 44[c,d]$, i.e., that $b \leq c$. Suppose, by way of contradiction, that $c < b$.  The only possible interval which satisfies the first conjunct is $[a,d]$. Specifically we have $b < d$. Now the interval $[c,b]$ falsifies the second conjunct. Conversely assume that $\mathcal{F} \models [a,b]\ii 34\cup\ii 44[c,d]$, i.e., $a<b \leq c<d$. Then the interval $z_i = [a,d]$ witnesses the first conjunct of the definition, and, for any $z_i$, if $z_i$ {\em starts} $y_i$, then it cannot {\em end} at $b$, proving that the second conjunct also holds. As for \underline{$\{\ip 2,\ii 14\}$} assume, first, that $\mathcal{F} \models \varphi([a,b],$ $[c,d])$. We wish to show that $\mathcal{F} \models [a,b]\ii 34\cup\ii 44[c,d]$, i.e., that $b \leq c$. Because of the first four conjuncts, there must be an interval {\em started by} $x_i$, and $z_p$ must be placed as the {\em ending} point of $[a,b]$, that is, $z_p=b$; then, $y_i$ must start at $z_p$ or after it, otherwise it would happen that either $y_i\ip 2 z_p$ or some $k_i$ {\em started} by $y_i$ is such that $k_i\ip 2 z_p$. On the other hand, assume that $\mathcal{F} \models [a,b]\ii 34\cup\ii 44[c,d]$, i.e., that $b \leq c<d$. In this case, it suffices to take $z_p=b$ and $z_i=[a,d]$ to satisfy all conjuncts. Finally, the $\ii34$-completeness of \underline{$\{\ip 0,\ip 2\}$} can be seen as follows:

\medskip

\begin{small}
\[
\begin{array}{llll}
x_i\ii 14 y_i &\leftrightarrow& \forall k_p(x_i\ip 0 k_p\leftrightarrow y_i\ip 0 k_p)\wedge\exists k_p(y_i\ip 2 k_p \wedge \neg(x_i\ip 2 k_p)) &  \{\ip0,\ip2\}
\end{array}
 \]
\end{small}

\noindent which is immediate to prove, as we stipulate that $x_i$ {\em starts} $y_i$ if and only if  $x_i$ and $y_i$ have the same points {\em before} them and there exists a point (the {\em ending} point of $x_i$) that it is {\em during} $y_i$ and not {\em during} $x_i$, and which allows us to obtain the result via the set $\{\ip 2,\ii 14 \}$.
\end{proof}

\medskip


The only non-symmetric interval-interval relation, $\ii 14$, does not present particularly difficult problems.

\begin{lem}\label{lem:tab:Ii14}
Tab.~\ref{tab:Ii14} is correct.
\end{lem}

\begin{proof}
The fact that the set \underline{$\{\ip 0,\ip 2\}$} is $\ii 14$-complete has already been
proved as a necessary step to delineate the $\mathsf{mcs}(\ii34)$. As for the set \underline{$\{\ip 0,\ip 4\}$}, it is easily verified that:

\medskip

\begin{small}
\begin{align*}
\begin{array}{llll}
x_i\ii 14 y_i &\leftrightarrow& \forall z_p(x_i\ip0 z_p\leftrightarrow y_i\ip 0 z_p)\wedge \exists z_p(x_i\ip 4 z_p \wedge \neg(y_i\ip 4 z_p)), & \{\ip 0,\ip 4\}
\end{array}
\tag*{\qEd}
\end{align*}
\end{small}

\def\popQED{}
\end{proof}

\medskip


\begin{lem}\label{lem:tab:Ii24}
Tab.~\ref{tab:Ii24} is correct.
\end{lem}

\begin{table}[t]
\small
\begin{center}
\begin{tabular}{|p{0.20\textwidth}|p{0.20\textwidth}|p{0.20\textwidth}|p{0.20\textwidth}|}
\hline
Proved & Symmetric & Implied & Deduction Chain\\
\hline
$\{\ip 2,<\}$                    &-                           & $\{\ip0,\ip2 \}$ & $\ii34$ \\
$\{\ip 0,\ii 44,\ii04\}$         &$\{\ip 4,\ii 44,\ii04\}$    & $\{\ip0,\ip3 \}$ & $\ii34$ \\
$\{\ip 1,\ii 44,\ii04\}$         &$\{\ip 3,\ii 44,\ii04\}$    & $\{\ip0,\ip4 \}$ & $\ii34$ \\
                                 &                            & $\{\ip0,\ii03\}$ & $\ii34$ \\
                                 &                            & $\{\ip1,\ip2 \}$ & $\ii34$ \\
                                 &                            & $\{\ip1,\ip3 \}$ & $\ii34$ \\
                                 &                            & $\{\ip1,\ip4 \}$ & $\ii34$ \\
                                 &                            & $\{\ip1,\ii03\}$ & $\ii34$ \\
                                 &                            & $\{\ip2,\ip3 \}$ & $\ii34$ \\
                                 &                            & $\{\ip2,\ip4 \}$ & $\ii34$ \\
                                 &                            & $\{\ip2,\ii14\}$ & $\ii34$ \\
                                 &                            & $\{\ip2,\ii03\}$ & $\ii34$ \\
                                 &                            & $\{\ip3,\ii14\}$ & $\ii34$ \\
                                 &                            & $\{\ip4,\ii14 \}$& $\ii34$ \\
                                 &                            & $\{\ii14,\ii03\}$& $\ii34$ \\
                                 &                            & $\{\ii14,\ii04\}$& $\ii34$ \\
                                 &                            & $\{\ii14,\ii44\}$& $\ii34$ \\
                                 &                            & $\{\ii03,\ii44\}$& $\ii34$ \\
                                 &                            & $\{\ii04,\ii 03\}$ & $\ii34$ \\
                                 &                            & $\{\ii34\}$ & Section~\ref{subsec:allimlin}\\
\hline
\end{tabular}
\caption{The spectrum of the $\mathsf{mcs}(\ii 24)$. - Class: $\Lin$.}(Lemma~\ref{lem:tab:Ii24}.)\label{tab:Ii24}
\end{center}
\end{table}

\begin{proof}
Consider the following definitions:

\medskip

\begin{small}
\[ x_i\ii 24 y_i \leftrightarrow \left\{ \begin{array}{ll}
                                   \exists z_p k_p(x_i \ip 2 z_p \wedge \forall t_p (x_i \ip 2 t_p \rightarrow t_p< k_p) & \{\ip 2,<\}\\
                                   \hspace{1cm} {} \wedge \forall s_p (\forall t_p (x_i \ip 2 t_p \rightarrow t_p< s_p) \rightarrow \neg (s_p < k_p))\\
                                    \hspace{1cm} {} \wedge y_i \ip 2 k_p \wedge \neg(y_i \ip 2 z_p)) & \\
                                    & \\
                                    \exists z_i t_i w_i (\forall k_p (x_i \ip 0 k_p \leftrightarrow z_i \ip 0 k_p) \wedge \forall k_p (y_i \ip 0 k_p \leftrightarrow w_i \ip 0 k_p) \wedge {} & \{\ip 0,\ii 44,\ii 04\}\\
                                    \hspace{0.7cm}  z_i \ii 44 t_i \wedge\exists s_i(w_i\ii 04 s_i\wedge\neg(y_i\ii 04 s_i))\wedge{} \\
                                    \hspace{0.7cm} \neg \forall k_p (x_i \ip 0 k_p \leftrightarrow y_i \ip 0 k_p) \wedge {}\\
                                    \hspace{0.7cm} \displaystyle{\bigwedge_{\substack{u_i, v_i \in \{w_i, x_i, y_i, z_i, t_i\}\\ \{u_i, v_i \} \neq \{z_i, t_i\}}}  (\neg(u_i \ii 04 v_i) \wedge \neg(u_i\ii 44 v_i))}) & \\
&\\
 \exists z_i t_i w_i (\exists k_p (x_i \ip 1 k_p \wedge z_i \ip 1 k_p) \wedge \exists k_p (y_i \ip 1 k_p \wedge w_i \ip 1 k_p) \wedge {} & \{\ip 1,\ii 44,\ii 04\}\\
                                    \hspace{0.7cm}  z_i \ii 44 t_i \wedge\exists s_i(w_i\ii 04 s_i\wedge\neg(y_i\ii 04 s_i))\wedge{} \\
                                    \hspace{0.7cm} \neg \exists k_p (x_i \ip 1 k_p \wedge y_i \ip 1 k_p) \wedge {}\\
                                   \hspace{0.7cm} \displaystyle{\bigwedge_{\substack{u_i, v_i \in \{w_i, x_i, y_i, z_i, t_i\}\\ \{u_i, v_i \} \neq \{z_i, t_i\}}}  (\neg(u_i \ii 04 v_i) \wedge \neg(u_i\ii 44 v_i))}) & \\
   \end{array} \right. \]
\end{small}

\medskip

\noindent Let us start proving the correctness of the case \underline{$\{\ip 2,<\}$}. Suppose that $\mathcal{F} \models [a,b]\ii 24 [c,d]$. Then, by definition, $a<c<b<d$; by setting $z_p=c$ and $k_p=b$ we satisfy all requirements. Conversely, suppose that $\mathcal{F} \models \varphi([a,b],[c,d])$. First, observe that $x_i=[a,b]$ must have and internal point $z_p$. Then, the only way to place $k_p$ is by having that $k_p$ is the smallest point which is greater than every internal point of $x_i$, that is, $k_p$ must {\em end} $x_i$. Now, $y_i$ must {\em contain} $k_p$ (so, it must {\em end} after it), but not $z_p$ (so, it must {\em start} after $z_p$ or on it). Consider, now, \underline{$\{\ip 1,\ii 44,\ii04\}$}. Suppose that $\mathcal{F} \models [a,b]\ii 24 [c,d]$. Then, by definition, $a<c<b<d$; by setting $z_i=[a,c]$, $w_i = [c,b]$, $t_i=[b,d]$, and $s_i=[a,d]$, we satisfy all requirements. Suppose, now, that $\mathcal{F} \models \varphi([a,b],[c,d])$. We prove that $[a,b]\ii 24[c,d]$ by eliminating every other possibility. First, observe that they cannot have their {\em starting} point in common (third line), and they cannot {\em contain} each other nor can one of them be {\em after} the other (fourth line). Thus, we eliminated $=_i,\ii 14,\ii 04,\ii 44$ among the possible relation between $x_i$ and $y_i$. Assume, by way of contradiction, that $y_i\ii34 x_i$. The second and fourth conjuncts together imply that $w_i \ii14 y_i$, and since $w_i$ (resp., $z_i$) has the same {\em starting} point as $y_i$ (resp., $x_i$), and since
$z_i\ii 44 t_i$, obviously $w_i\ii 44 t_i$, which is in contradiction with the fourth line. Assume now that $x_i\ii34 y_i$. First, observe that
$z_i$ cannot be shorter than $x_i$, otherwise it would be $z_i\ii 44 w_i$ (contradicting the fourth line); but then, $z_i$ cannot be equal to $x_i$
or longer than it, as in that case $x_i\ii 44 t_i$ (contradicting, once again, the requirement in the fourth line). Now, suppose that $x_i\ii 03 y_i$.
Let $w_i=[e,f]$, where $e=c$. Clearly, to comply with the fourth line, it cannot be that $f<a$ or that $f>b(=d)$. If, on the other hand, $a\leq f\leq b$, then it is impossible to place $z_i=[g,h]$ ($g=a$), as it must be that $h<f$ (since $w_i\ii 44 t_i$ leads to a contradiction), but, then, $z_i\ii 04 y_i$, in
contradiction with the fourth line. Therefore, we cannot place $w_i$, which means that we have a contradiction. Next suppose that $y_i\ii 03 x_i$. Let $w_i=[e,f]$, where $e=c$; clearly, $f\ge d$, as, otherwise, we would have that $x_i$ {\em contains} $w_i$ (forbidden by the fourth line); but, then, it is impossible to find any $s_i$ that {\em contains} $w_i$ and not $y_i$ (second line), leading to a contradiction. Finally, if $y_i\ii 24 x_i$, it is impossible to place $w_i,z_i$ and $t_i$ in such a way that $z_i\ii 44 t_i$, $\neg(z_i\ii 04 w_i)$, and $\neg(w_i\ii 04 t_i)$. Having eliminated every other possibility,
it must be that $x_i\ii 24 y_i$. Finally, observe that he case \underline{$\{\ip 0,\ii 44,\ii04\}$} can be dealt with the very same argument, only having $\ip 0$ playing the same role as $\ip 1$. The corresponding definition is therefore almost identical to the one of the previous case.
\qedhere

%

\end{proof}




\begin{lem}\label{lem:tab:Ii04}
Tab.~\ref{tab:Ii04} is correct.
\end{lem}

\begin{proof} Consider the following definitions:

\medskip
\begin{small}
\[ x_i\ii 04 y_i \leftrightarrow \left\{ \begin{array}{ll}
\exists z_i t_iz_p t_p x_i' x_i'' (y_i\ip 2 z_p\wedge y_i\ip 2 t_p\wedge z_i\ii 24 t_i\wedge x_i'\ii 44 x_i''\wedge z_i \ip 2 z_p \wedge t_i \ip 2 t_p\wedge & \{\ip 2,\ii 24,\ii 44\}\\
                                                     \hspace{2cm}\neg(z_i\ip 2 t_p)\wedge\neg(t_i\ip 2 z_p)\wedge & \\
                                                     \hspace{2cm}\bigwedge\limits_{s_i\in\{x_i,x_i',x_i''\}}\hspace{-0.5cm}(\neg(s_i\ip 2 z_p)\wedge\neg(s_i\ip 2 t_p))\wedge & \\
                                                     \hspace{2cm}\bigwedge\limits_{s_i\in\{y_i,x_i,z_i,t_i\},r\in\{\ii 44,\ii 24\}}\hspace{-1.5cm}(\neg(x_i'~r~s_i)\wedge\neg(s_i~r~x_i'))\wedge & \\
                                                     \hspace{2cm}\bigwedge\limits_{s_i\in\{y_i,x_i',x_i'',z_i,t_i\},r\in\{\ii 44,\ii 24\}}\hspace{-1.8cm}(\neg(x_i~r~s_i)\wedge\neg(s_i~r~x_i))\wedge & \\
                                                     \hspace{2cm}\bigwedge\limits_{s_i\in\{y_i,x_i,z_i,t_i\},r\in\{\ii 44,\ii 24\}}\hspace{-1.5cm}(\neg(x_i''~r~s_i)\wedge\neg(s_i~r~x_i'')))& \\
&\\
                                    \exists z_i t_i(\forall k_p(x_i \ip 0 k_p \leftrightarrow z_i \ip 0 k_p) \wedge \forall k_p (y_i\ip 0 k_p\leftrightarrow t_i\ip 0 k_p)\wedge {} & \{\ip 0,\ii 24,\ii 44\}\\
                                    \hspace{0.8cm}\forall w_i(y_i\ii 44 w_i \leftrightarrow z_i\ii 44 w_i)\wedge t_i \ii 24 z_i \wedge {}& \\
                                    \hspace{0.8cm}\neg \forall k_p (x_i \ip 0 k_p \leftrightarrow y_i \ip 0 k_p)\wedge &\\
                                    \hspace{0.8cm}\displaystyle{\bigwedge\limits_{\substack{s_i, w_i \in \{x_i, y_i, z_i, t_i\}\\ \{s_i, w_i \} \neq \{t_i, z_i \}}}\hspace{-0.7cm}(\neg(s_i\ii 24 w_i)\wedge\neg(s_i\ii 44 w_i))}) & \\
&\\
\exists z_i t_i(\forall k_p(x_i \ip 1 k_p \leftrightarrow z_i \ip 1 k_p) \wedge \forall k_p (y_i\ip 1 k_p\leftrightarrow t_i\ip 1 k_p)\wedge {} & \{\ip 1,\ii 24,\ii 44\}\\
                                    \hspace{0.8cm}\forall w_i(y_i\ii 44 w_i \leftrightarrow z_i\ii 44 w_i)\wedge t_i \ii 24 z_i \wedge {}& \\
                                    \hspace{0.8cm}\neg \exists k_p (x_i \ip 1 k_p \wedge y_i \ip 1 k_p) \wedge & \\
                                    \hspace{0.8cm}\displaystyle{\bigwedge\limits_{\substack{s_i, w_i \in \{x_i, y_i, z_i, t_i\}\\ \{s_i, w_i \} \neq \{t_i, z_i \}}}\hspace{-0.7cm}(\neg(s_i\ii 24 w_i)\wedge\neg(s_i\ii 44 w_i))})
\end{array} \right. \]
\end{small}

\medskip

\noindent Let us start by \underline{$\{\ip 2,\ii 24,\ii 44\}$}. Suppose that $\mathcal{F} \models [c,d]\ii 04 [a,b]$. Then, by definition, $a<c<d<b$; to satisfy $\varphi$ we set: $x_i'=[a,c],x_i''=[d,b],z_i=[a,d],t_i=[c,b],z_p=c$, and $t_p=d$. Suppose, now, that $\mathcal{F} \models \varphi([c,d],[a,b])$ and that $z_i = [b_z, e_z]$, $t_i = [b_t, e_t]$, $x' = [b', e']$ and $x'' = [b'', e'']$. First, observe that $y_i= [a,b]$ must have two internal points $z_p$ and $t_p$, and that $z_p$ ($t_p$) must be {\em contained} in  $z_i$ ($t_i$) and not in $t_i$ ($z_i$).  Since $z_i$ {\em overlaps} $t_i$ we must have  $z_p < t_p$. We deduce that $a < z_p \leq b_t < e_z \leq t_p < b$. The constraints on $x'$ and $x''$ in lines 1, 3, 4 and 6 of the definition imply that we must have $e' = b_t = z_p$ and $b'' = e_z = t_p$. This gives $a < z_p = b_t = e' < b'' = e_z = t_p < b$. Now, the only way to place $x_i = [c,d]$  without violating the third and fifth line of the definition is that $e'= c < d = b''$. Consequently $a < c < b < d$, as desired. Consider, now, \underline{$\{\ip 1,\ii 24,\ii 44\}$}. Suppose that $\mathcal{F}\models [c,d]\ii 04 [a,b]$. Then, by definition, $a<c<d<b$; by setting $z_i=[c,b]$ and $t_i=[a,d]$, we satisfy all requirements. Suppose, now, that $\mathcal{F} \models \varphi([c,d],[a,b])$. Once more, we eliminate every possible relation between $x_i$ and $y_i$, except for $x_i\ii 04 y_i$. As before, $\ii 14$, $\ii 24$, and $\ii 44$ are immediately eliminated (third and fourth line). If $x_i\ii 34 y_i$, then, since $z_i$ (resp., $t_i$) and $x_i$ (resp., $y_i$) have the same {\em starting} point (first line), it is not possible to have $t_i\ii 24 z_i$, contradicting the second line. Suppose, now, that $y_i\ii 34 x_i$. First, observe that $z_i$ (which {\em starts} at the same point as $x_i$) cannot {\em start} or be {\em equal to} $x_i$, since in that case $t_i\ii 24 x_i$ (contradicting the fourth line). But then, if $z_i$ is {\em started by} $x_i$, there will be some interval that is {\em after} $y_i$ which is not {\em after}
$z_i$ itself, contradicting the second line. Thus, $\ii 34$ is eliminated.  Assume that $x_i\ii 03 y_i$. As before,
$z_i$ cannot {\em start} or be {\em equal to} $x_i$, since in that case $t_i\ii 24 x_i$ (contradicting the fourth line). Moreover, if $z_i$ is {\em started by} $x_i$, it would imply that $y_i\ii 24 z_i$, again, a contradiction with the fourth line. If $y_i\ii 03 x_i$, $y_i=_i x_i$, or $y_i\ii 04 x_i$, then it is impossible to correctly place $z_i$ and $t_i$ in such a way that $t_i\ii 24 z_i$. Having eliminated in this case all other possibilities, we conclude that $x_i\ii 04 y_i$. Once again, we can deal with \underline{$\{\ip 0,\ii 24,\ii 44\}$} with the exact same argument, where $\ip 0$ plays the role of $\ip1 $.
\qedhere
\end{proof}

\begin{table}[t]
\small
\begin{center}
\begin{tabular}{|p{0.20\textwidth}|p{0.20\textwidth}|p{0.20\textwidth}|p{0.20\textwidth}|}
\hline
Proved & Symmetric & Implied & Deduction Chain\\
\hline
$\{\ip 0,\ii 24,\ii44\}$         &$\{\ip 4,\ii 24,\ii44\}$    & $\{\ip0, \ip2 \}$ & $\ii 34$   \\
$\{\ip 1,\ii 24,\ii44\}$         &$\{\ip 3,\ii 24,\ii44\}$    & $\{\ip0, \ip3 \}$ & $\ii 34$   \\
$\{\ip 2,\ii 24,\ii44\}$         &-                           & $\{\ip0, \ip4 \}$ & $\ii 34$   \\
                                 &                            & $\{\ip0, \ii03\}$ & $\ii 34$   \\
                                 &                            & $\{\ip1, \ip2 \}$ & $\ii 34$   \\
                                 &                            & $\{\ip1, \ip3 \}$ & $\ii 34$   \\
                                 &                            & $\{\ip1, \ip4 \}$ & $\ii 34$   \\
                                 &                            & $\{\ip1, \ii03 \}$ & $\ii 34$   \\
                                 &                            & $\{\ip2, \ip3 \}$ & $\ii 34$   \\
                                 &                            & $\{\ip2, \ip4 \}$ & $\ii 34$   \\
                                 &                            & $\{\ip2, \ii14 \}$ & $\ii 34$   \\
                                 &                            & $\{\ip2, \ii03 \}$ & $\ii 34$   \\
                                 &                            & $\{\ip2, \ii44,< \}$ & $\ii 24$ \\
                                 &                            & $\{\ip3, \ii14 \}$ & $\ii 34$   \\
                                 &                            & $\{\ip4, \ii14 \}$ & $\ii 34$   \\
                                 &                            & $\{\ii14, \ii24 \}$ & $\ii 34$   \\
                                 &                            & $\{\ii14, \ii03 \}$ & $\ii 34$   \\
                                 &                            & $\{\ii14, \ii44 \}$ & $\ii 34$   \\
                                 &                            & $\{\ii24, \ii03 \}$ & $\ii 34$   \\
                                 &                            & $\{\ii03, \ii44 \}$ & $\ii 34$   \\
                                 &                            & $\{\ii34 \}$ & Section~\ref{subsec:allimlin}\\
\hline
\end{tabular}
\caption{The spectrum of the $\mathsf{mcs}(\ii 04)$. - Class: $\Lin$.}(Lemma~\ref{lem:tab:Ii04}.)\label{tab:Ii04}
\end{center}
\end{table}


\begin{lem}\label{lem:tab:Ii44}
Tab.~\ref{tab:Ii44} is correct.
\end{lem}

\begin{proof}
As always, we start with the definition:

\medskip

\begin{small}
\[ \begin{array}{llll}x_i\ii 44 y_i &\leftrightarrow&
                                    \exists z_it_i\big( \exists k_p(x_i\ip 1 k_p \wedge z_i \ip 1 k_p) \wedge \neg \exists k_p (x_i \ip 1 k_p \wedge y_i \ip 1 k_p) \wedge  &  \{\ip 1,\ii 24,\ii 04\}\\
                                    &&\hspace{0.8cm}\forall w_i(\forall k_p(y_i\ip 1 k_p\rightarrow w_i\ip 1 k_p)\rightarrow\neg(z_i\ii 04 w_i))\wedge {}& \\
                                    &&\hspace{0.8cm}\forall w_i(\forall k_p(t_i\ip 1 k_p\rightarrow w_i\ip 1 k_p)\rightarrow\neg(w_i\ii 04 x_i))\wedge {}& \\
                                    &&\hspace{0.8cm}\forall w_i(\forall k_p(y_i\ip 1 k_p\rightarrow w_i\ip 1 k_p)\rightarrow\neg(z_i\ii 24 w_i))\wedge {}& \\
                                    &&\hspace{0.8cm}\forall w_i(\forall k_p(y_i\ip 1 k_p\rightarrow w_i\ip 1 k_p)\rightarrow\neg(x_i\ii 04 w_i))\wedge {}& \\
                                    &&\hspace{0.8cm}\forall w_i(\forall k_p(y_i\ip 1 k_p\rightarrow w_i\ip 1 k_p)\rightarrow\neg(w_i\ii 24 z_i))\wedge {}& \\
                                    &&\hspace{0.8cm} z_i\ii 24 t_i\wedge \displaystyle{\bigwedge\limits_{\substack{s_i,w_i \in \{x_i,y_i,z_i,t_i\} \\  \{ s_i, w_i\} \neq \{ z_i, t_i \}}}\hspace{-0.5cm}(\neg(s_i\ii 24 w_i)\wedge\neg(s_i\ii 04 w_i))})
   \end{array}  \]
\end{small}

\medskip

\begin{table}[t]
\small
\begin{center}
\begin{tabular}{|p{0.20\textwidth}|p{0.20\textwidth}|p{0.20\textwidth}|p{0.20\textwidth}|}
\hline
Proved & Symmetric & Implied & Deduction Chain\\
\hline
$\{\ip 0,\ii 24,\ii04\}$         &$\{\ip 4,\ii 24,\ii04\}$    & $\{\ip0, \ip2 \}$ & $\ii 34$   \\
$\{\ip 1,\ii 24,\ii04\}$         &$\{\ip 3,\ii 24,\ii04\}$    & $\{\ip0, \ip3 \}$ & $\ii 34$   \\
                                 &                            & $\{\ip0, \ip4 \}$ & $\ii 34$   \\
                                 &                            & $\{\ip0, \ii03\}$ & $\ii 34$   \\
                                 &                            & $\{\ip1, \ip2\}$  & $\ii 34$   \\
                                 &                            & $\{\ip1, \ip3 \}$ & $\ii 34$   \\
                                 &                            & $\{\ip1, \ip4 \}$ & $\ii 34$   \\
                                 &                            & $\{\ip1, \ii03 \}$ & $\ii 34$   \\
                                 &                            & $\{\ip2, \ip3 \}$  & $\ii 34$   \\
                                 &                            & $\{\ip2, \ip4 \}$  & $\ii 34$   \\
                                 &                            & $\{\ip2, \ii14 \}$ & $\ii 34$   \\
                                 &                            & $\{\ip2, \ii03 \}$  & $\ii 34$   \\
                                 &                            & $\{\ip3, \ii14 \}$  & $\ii 34$ \\
                                 &                            & $\{\ip4, \ii14 \}$  & $\ii 34$   \\
                                 &                            & $\{\ii14, \ii24 \}$ & $\ii 34$   \\
                                 &                            & $\{\ii14, \ii03 \}$ & $\ii 34$   \\
                                 &                            & $\{\ii14, \ii04 \}$ & $\ii 34$   \\
                                 &                            & $\{\ii24, \ii03 \}$ & $\ii 34$   \\
                                 &                            & $\{\ii03, \ii04 \}$ & $\ii 34$   \\
                                 &                            & $\{\ii34 \}$        & Section~\ref{subsec:allimlin}   \\
\hline
\end{tabular}
\caption{The spectrum of the $\mathsf{mcs}(\ii 44)$. - Class: $\Lin$.} (Lemma~\ref{lem:tab:Ii44}.)\label{tab:Ii44}
\end{center}
\end{table}

\noindent Following the same schema of the previous two relations, consider the case of \underline{$\{\ip 1,\ii 24,\ii 04\}$}, and
suppose that $\mathcal{F} \models [a,b]\ii 44 [c,d]$. Then, by definition, $a<b<c<d$; by setting $z_i=[a,c]$ and $t_i=[b,d]$, all requirements
are satisfied. Suppose, now, that $\mathcal{F} \models \varphi([a,b],[c,d])$. We prove that $[a,b]\ii 44[c,d]$ by eliminating every other possible relation
that may hold between $x_i$ and $y_i$. The relations $=_i$, $\ii 14$, $\ii 24$ and $\ii 04$ are eliminated immediately thanks to the first and last lines of the definition. Suppose, by way of contradiction, that $x_i\ii 34 y_i$. The interval $z_i$, which must {\em start} at $a$ (first line), cannot end between $c$ and $d$ in order to comply with the fact that it cannot {\em overlap} $y_i$ (last line). If $z_i$ {\em ends} at $d$ then, since $z_i \ii 24 t_i$, we must have (i) $t_i$ beginning between $a$ and $b$ which causes it to contain $y_i$ contradicting the last line, or (ii) $t_i$ beginning at $b = c$ which causes a contradiction with the fourth line, or (iii)  $t_i$ beginning between at $c$ and $d$ causing $y_i$ to {\em overlap} $t_i$, contradicting the last line. If $z_i$ {\em ends} after $d$, then $y_i$ is {\em contained} in $z_i$ contradicting the last line. If $z_i$ {\em ends} at $c$  then $t_i$ {\em overlaps} $x_i$ contradicting the last line. If $z_i$ {\em ends} between $a$ and $b$ the overlap between $z_i$ and $t_i$ begins with $t_i$ and is contained in $x_i$, contradicting the third line. This eliminates all possibilities for the placement of $z_i$ and hence the case $x_i\ii 34 y_i$. Suppose, now, that $y_i\ii 34 x_i$. Then there exists an interval starting with $y_i$ which {\em overlaps} $z_i$, contradicting the sixth line.  Thus, $\ii 34$ is eliminated. Similarly, if $x_i\ii 03 y_i$, then there exists an interval starting with $y_i$ which {\em overlaps} $z_i$, contradicting the sixth line. Suppose next that $y_i\ii 03 x_i$. Note that $z_i$ must end after $x_i$ as the overlap between $z_i$ and $t_i$ cannot be {\em contained} in $x_i$ (third line) and $x_i$ cannot overlap $t_i$ (last line). This means that $z_i$ {\em contains} $y_i$, contradicting once more the last line. Thus, $\ii 03$ is eliminated as well. Finally, if $y_i\ii 44 x_i$, it is impossible to correctly place $z_i$ and $t_i$. Having eliminated every other possibility, we conclude that $x_i\ii 44 y_i$. As in the proofs of the previous two relations,  the same argument solves the case of \underline{$\{\ip 0,\ii 24,\ii 04\}$}:

\begin{small}
  \begin{align*}
    \begin{array}[b]{llll}x_i\ii 44 y_i &\leftrightarrow&
                                    \exists z_it_i\big( \forall k_p(x_i\ip 0 k_p \leftrightarrow z_i \ip 0 k_p) \wedge \neg \forall k_p (x_i \ip 0 k_p \leftrightarrow y_i \ip 0 k_p) \wedge  &  \{\ip 0,\ii 24,\ii 04\}\\
                                    &&\hspace{0.8cm}\forall w_i(\forall k_p(y_i\ip 0 k_p\rightarrow w_i\ip 0 k_p)\rightarrow\neg(z_i\ii 04 w_i))\wedge {}& \\
                                    &&\hspace{0.8cm}\forall w_i(\forall k_p(t_i\ip 0 k_p\rightarrow w_i\ip 0 k_p)\rightarrow\neg(w_i\ii 04 x_i))\wedge {}& \\
                                    &&\hspace{0.8cm}\forall w_i(\forall k_p(y_i\ip 0 k_p\rightarrow w_i\ip 0 k_p)\rightarrow\neg(z_i\ii 24 w_i))\wedge {}& \\
                                    &&\hspace{0.8cm}\forall w_i(\forall k_p(y_i\ip 0 k_p\rightarrow w_i\ip 0 k_p)\rightarrow\neg(x_i\ii 04 w_i))\wedge {}& \\
                                    &&\hspace{0.8cm}\forall w_i(\forall k_p(y_i\ip 0 k_p\rightarrow w_i\ip 0 k_p)\rightarrow\neg(w_i\ii 24 z_i))\wedge {}& \\
                                    &&\hspace{0.8cm} z_i\ii 24 t_i\wedge \displaystyle{\bigwedge\limits_{\substack{s_i,w_i \in \{x_i,y_i,z_i,t_i\} \\  \{ s_i, w_i\} \neq \{ z_i, t_i \}}}\hspace{-0.5cm}(\neg(s_i\ii 24 w_i)\wedge\neg(s_i\ii 04 w_i))})
   \end{array}
\tag*{\qEd}
\end{align*}
\end{small}
\def\popQED{}
\end{proof}

\medskip

\section{Incompleteness Results in The Class \texorpdfstring{$\Lin$}{Lin} and The Class \texorpdfstring{$\Dis}{Dis}$}\label{sec:lin-incomp}

We now describe and prove the `other half' of the picture, by identifying all maximally incomplete sets for each relation $r\in\allr^+$. To treat incompleteness, since most incomplete set appears as $\mathsf{MIS}$ for more than one relation, we present these results as follows. In the table we list, in the leftmost column the maximal incomplete sets, and in the topmost row all $\allr^+$-relations. Whenever the crossing point of a column and a row is marked, the set corresponding to its row is a $\mathsf{MIS}$ for the relation corresponding to its column, and that fact will be justified in the proof. The section at the top contains those sets for which we give an explicit construction, while the section at the bottom contains all symmetric results. Finally, the sets in the topmost part which are symmetric to themselves will be proved $r$-incomplete up to reverse of relations (i.e., their $r'$-incompleteness relative to the reverse of a given $r$ is implied, but not mentioned). The results for the sub-languages induced by $\alli^+$ and $\allm^+$, which turn out to be much simpler than those for $\allr^+$, are included in the next section. As already mentioned, we shall prove specific undefinability results via surjective truth-preserving relations between models, and we shall be therefore forced to choose a specific class of linearly ordered sets. At the end of this analysis, it shall turn out that the expressive power in the classes $\Lin$ and $\Dis$ is identical (every counterexample for the class $\Lin$ is based on a discrete structure) while $\Den$ and $\Unb$ are different from the former two, and (slightly) different from each other. The class $\Fin$ (treated in Part II, along with $\Den$ and $\Unb$), shall require, as we shall see, a deeper analysis. Throughout the following proofs, let $\mathcal{F}$ and $\mathcal{F'}$ be two concrete structures, and, for every given case, let $S$ be a set of relations that we claim to be $r$-incomplete for a given $r$. Case-by-case, we shall define two domains $\mathbb D,\mathbb D'$ and a surjective $S$-relation between $\mathcal{F}$ and $\mathcal{F'}$ that breaks $r$. Notation- and terminology-wise, most of the constructions are based on the {\em same} domain (in terms of
their elements and relative ordering): we distinguish them by means of the superscript $'$, and, with an abuse of terminology, we use the term {\em identity} relation (over points or over intervals), denoted by $Id_p$ or $Id_i$ to indicate the relation that respects the element's name (i.e., it relates $a$ with $a'$, $[a,b]$ with $[a',b']$, and so on). Notice that, all together, the results presented here imply those reported in Section~\ref{sec:harvest}, which
will therefore require no formal proof. Notice also that each set $S$ listed as $r$-incomplete for some $r$ in this section has been obtained by means of the technique shown in Fig.~\ref{fig:undef}. This means that it must be maximal by definition: if that was not the case, then $S\cup\{r'\}$ would be $r$-incomplete for some $r'\notin S$, but $S\cup\{r'\}$ contains some $r$-complete set, which is a contradiction.

\medskip

In what follows, two cases are particularly difficult, and they require an additional concept defined here.

\begin{defi}\label{def:iterated}
We say that the structure $\langle\mathbb D^*,\prec \rangle$ is {\em iterated discrete} if and only if:

\begin{enumerate}[label={(\emph{\roman*})}]
\item $\mathbb D^* = \mathbb Z\times\mathbb Q$;
\item $\prec$ is defined as follow: $(n,q)\prec(m,r)$ if and only if (1) $q<r$, or (2) $q=r$ and $n<m$, that is, it defines a linear order as the reverse lexicographic order between pairs.
\end{enumerate}
\end{defi}

\noindent Notice that an iterated discrete structure is, in fact, discrete: the direct predecessor (resp., successor) or $(n,q)$ is $(n-1,q)$ (resp., $(n+1,q)$). The two cases that require a construction based on an iterated discrete structure (one in the sub-language induced by $\alli^+$, the other one in the general case of $\allr^+$) could be sorted out by a similar (and simpler)  construction based on $\mathbb{Q}$; however, basing it on a discrete structure yields the stronger incompleteness result we want, that applies to both $\Dis$ and $\Lin$.

\subsection{(Maximal) $\alli^+$- and $\allm^+$-Incompleteness}

As we have done while analyzing the definability of relations, we focus our attention, first, on the sub-languages induced by $\alli^+$ and by $\allm^+$.

\begin{table}[t]
\small
\begin{center}
\begin{tabular}{p{0.34\textwidth}|p{0.04\textwidth}|p{0.04\textwidth}|p{0.04\textwidth}|p{0.04\textwidth}|p{0.04\textwidth}|p{0.04\textwidth}|p{0.04\textwidth}|}
{\em Proved}                     & $=_i$           & $\ii34$          & $\ii14$          & $\ii03$          & $\ii24$          & $\ii04$          & $\ii44$           \\
\hline
$\{\ii 44, \ii 04, \ii 24\}$     &$\bullet$        &                  &                  &                  &                  &                  &
\\
\hline
$\{\ii 44, \ii 04, \ii 24,=_i\}$ &                 &$\bullet$         &$\bullet$         &$\bullet$         &                  &                  &
\\
\hline
$\{\ii 03,=_i\}$                 &                 &$\bullet$         &$\bullet$         &                  &$\bullet$         &$\bullet$         &$\bullet$
\\
\hline
\\
{\em Symmetric}                     & $=_i$           & $\ii34$          & $\ii14$          & $\ii03$          & $\ii24$          & $\ii04$          & $\ii44$           \\
\hline
$\{\ii 14,=_i\}$                 &                 &$\bullet$         &                     & $\bullet$        &$\bullet$         &$\bullet$         &$\bullet$
\\
\hline
\end{tabular}                                                                                                                                                                                                                                         \caption{$\mathsf{MIS}(r)$, for each $r\in\alli^+$; upper part: sets for which we give an explicit construction; lower part: symmetric ones. - Classes: $\Lin$, $\Dis$.}(Lemma~\ref{lem:tab:MISIlin}.)\label{tab:MISIlin}
\end{center}
\end{table}

\begin{lem}\label{lem:tab:MISIlin}
Tab.~\ref{tab:MISIlin} is correct and complete.
\end{lem}

\begin{proof}
The fact that \underline{$\{\ii 44, \ii 04, \ii 24\}$} is maximally incomplete for $=_i$ is justified as follows.
Let $\mathbb D= \mathbb D'=\{a<b<c\}$, and $\zeta=\zeta_i=\{([a,c],[a,b])\}\cup Id_i$. We have that $\zeta$ is total and the interval-interval relations in $S$ are trivially respected. But $[a,b]$ is the image of both $[a,b]$ and $[a,c]$, proving that $=_i$ is broken. Now, we prove that \underline{$\{\ii 44, \ii 04, \ii 24,=_i\}$} is $\ii34$-,$\ii 14$-, and $\ii 03$-incomplete. To this end, we take $\mathbb D = \mathbb D' =\{a<b<c\}$ and $\zeta=\zeta_i=\{([a,c],[a',c']),([a,b],[b',c']),([b,c],[a',b'])\}$. The relation $\zeta:\mathcal F\rightarrow\mathcal F'$ is clearly surjective (actually, it is a bijection, so equality between intervals is preserved) with respect to interval, and every relations in $S$ is trivially respected. However, it clearly brakes all intended relations. Consider, now, the case of \underline{$\{\ii 03,=_i\}$}, which we have to prove that it is $\ii 34$-,$\ii 14$-,$\ii 24$-,$\ii 04$-,$\ii 44$-incomplete. We proceed as follows: first, we give a construction that proves this set to be $\ii 34$-,$\ii 14$-,$\ii 24$-, and $\ii 44$-incomplete, and, then, give a second construction that proves it to be $\ii 04$-incomplete. As for the first step, define $\mathcal F$ and $\mathcal F'$ both based on an iterated discrete structure $\mathbb D^*$ (see Def.~\ref{def:iterated}). Define a relation $\zeta=\zeta_i$ between them as follows: $\zeta_i$ is the union of $\{([(n,q),(m,r)],[(n',q'),(m',r')]) \mid q = r \mbox{ and } n, m \in  \mathbb{Z}\}$ and $\{ ([(n,q),(m,r)],[(n',q'-|r'-q'|),(m',r')]) \mid q < r \mbox{ and } n, m \in  \mathbb{Z}\}$. In this way, finite intervals (i.e., those containing a finite number of points) are related to themselves, while infinite ones are related to intervals with the same {\em ending} point but twice the length (with respect the rational coordinate). We want to prove that $\zeta$ is an $S$-relation. Obviously, the only interesting cases are those that involve at least one infinite interval. Consider two intervals $[(n,q),(m,r)]$ and $[(l,s),(m,r)]$ with $(l,s) \prec (n,q)$; they are $\ii{0}{3}$-related, so either both are infinite or $[(n,q),(m,r)]$ is finite and $[(l,s),(m,r)]$ is infinite, so under $\zeta$ the endpoints are kept constant and either both are doubled in length or $[(n,q),(m,r)]$ is kept fixed and only $[(l,s),(m,r)]$ is doubled. Either way, $\ii{0}{3}$ is respected. Since $\zeta$ is a bijection by definition, equality between intervals is respected too, and therefore $\zeta$ is a proper $S$-relation. It is clear, on the other hand, that $\ii{3}{4}, \ii{14}, \ii{2}{4}$ and $\ii{4}{4}$ are broken. Now, as for the second step, consider once again $\mathcal F$ and $\mathcal F'$ both based on an iterated discrete structure $\mathbb D^*$ and define a new $S$-relation $\zeta=\zeta_i$ between them as follows: $\zeta_i$ is the union of $\{([(n,q),(m,r)],[(n',q'),(m',r')]) \mid q = r \mbox{ and } n, m \in  \mathbb{Z}\}$ and $\{ ([(n,q),(m,r)],[(n',q'+(|r'-q'|/2)),(m',r')]) \mid q < r \mbox{ and } n, m \in  \mathbb{Z}\}$. In this way, finite intervals are related to themselves, while infinite ones are related to intervals with the same {\em ending} point but {\em half} the length (with respect to the rational coordinate). It is straightforward to check that this is a bijective $S$-relation which breaks $\ii{0}{4}$ and respects both equality between intervals and the relation $\ii 03$. The fact that the table is complete is a consequence of the following observation: every set not listed is either contained in some of the listed ones (and, thus, it is not maximal with respect to incompleteness) or $\alli^+$-complete (from the results of Section~\ref{sec:lin}).
\end{proof}

\begin{table}[t]
\small
\begin{center}
\begin{tabular}{p{0.34\textwidth}|p{0.04\textwidth}|p{0.04\textwidth}|p{0.04\textwidth}|p{0.04\textwidth}|p{0.04\textwidth}|}
{\em Proved}                     & $\ip 0$           & $\ip 1$          & $\ip 2$          & $\ip 3$          & $\ip 4$ \\
\hline
$\{\ip 3,\ip 4\}$                &$\bullet$          &$\bullet$         &$\bullet$         &                  &          \\
\hline
$\{\ip 0,\ip 2,\ip 4\}$          &                   &$\bullet$         &                  & $\bullet$        &          \\
\hline
\\
{\em Symmetric}                  & $\ip 0$           & $\ip 1$          & $\ip 2$          & $\ip 3$          & $\ip 4$ \\
\hline
$\{\ip 0,\ip 1\}$                &                   &                  &$\bullet$         &$\bullet$         &$\bullet$\\
\hline
\end{tabular}                                                                                                                                                                                                                                         \caption{$\mathsf{MIS}(r)$, for each $r\in\allm^+$; upper part: sets for which we give an explicit construction; lower part: symmetric ones. - Classes: $\Lin$, $\Dis$.}(Lemma~\ref{lem:tab:MISMlin}.)\label{tab:MISMlin}
\end{center}
\end{table}

\begin{lem}\label{lem:tab:MISMlin}
Tab.~\ref{tab:MISMlin} is correct and complete.
\end{lem}

\begin{proof} Let us start by proving that \underline{$\{\ip 3,\ip 4\}$} is $\ip0$-,$\ip1$-, and $\ip2$-incomplete. Take $\mathbb D=\mathbb D'=\{a<b<c\}$, $\zeta_p=Id_p$, and $\zeta_i=\{([b,c],[a',c']),([a,c],[b',c']),$ $([a,b],[a',b'])\}$. It is easy to verify that $\ip{3}$, $\ip{4}$ are respected, but none of the other mixed relations are. Then, to prove that \underline{$\{\ip 0,\ip2,\ip 4\}$} is $\ip1$- and $\ip3$-incomplete, we simply take $\mathbb D=\mathbb D'=\{a<b\}$, $\zeta_p=\{(a,b'),(b,a')\}$, and $\zeta_i=Id_i$. \end{proof}

\subsection{(Maximal) $\allr^+$-Incompleteness}

To conclude this section, we analyze the results shown in Tab.~\ref{tab:MISRlin}. Notice that some of the constructions presented here can be actually considered as generalizations of those seen above.

\begin{table}[t]
\small
\begin{center}
\begin{tabular}{p{0.34\textwidth}|p{0.022\textwidth}|p{0.022\textwidth}|p{0.022\textwidth}|p{0.022\textwidth}|p{0.022\textwidth}|p{0.022\textwidth}|p{0.022\textwidth}|p{0.022\textwidth}|p{0.022\textwidth}|p{0.022\textwidth}|p{0.022\textwidth}|p{0.022\textwidth}|p{0.022\textwidth}|p{0.022\textwidth}|p{0.022\textwidth}|}
{\em Proved}                                        & $=_p$             & $=_i$             & $<$              & $\ip0$           & $\ip1$           & $\ip2$           & $\ip3$           & $\ip4$           & $\ii34$          & $\ii14$          & $\ii03$          & $\ii24$          & $\ii04$          & $\ii44$           \\
 \hline
$\{\ip 0, \ip 2, \ip 4 \}\cup\alli^{+}$             & $\bullet$         &                   &          &                  &                  &                  &                  &                  &                  &                  &                  &                  &                  &                   \\
 \hline
$\{=_p, <, \ip 0,\ip 1, \ii 44, \ii 04, \ii 24\}$                       &                   &$\bullet$           &                 &                  &                  &                  &                  &                  &                  &                  &                  &                  &                  &                   \\
\hline
$\{=_p, <, \ip2, \ii 44, \ii 04, \ii 24\}$                       &                   &$\bullet$           &                 &                  &                  &                  &                  &                  &                  &                  &                  &                  &                  &                   \\
\hline
$\{=_p, =_i, \ip 1, \ii 14\}$                       &                   &                  &$\bullet$         &$\bullet$         &                  &                  &                  &                  &                  &                  &                  &                  &                  &                   \\
\hline
$\{=_p, \ip 1, \ii 44, \ii 04, \ii 24\}$            &                &   &$\bullet$         &$\bullet$         &                  &                  &                  &                  &                  &                  &                  &                  &                  &                   \\
\hline
$\{=_p, \ip 0, \ip 2, \ip 4\}\cup\alli^{+}$         &                 &  &$\bullet$         &                  &$\bullet$         &                  &$\bullet$         &                  &                  &                  &                  &                  &                  &                   \\
\hline
$\{=_p,\ip 2\}\cup\alli^+$                          &                 &  &                  &$\bullet$         &                  &                  &                  & $\bullet$        &                  &                  &                  &                  &                  &                   \\
\hline
$\{=_p,=_i,<,\ip 2,\ii 44, \ii 04, \ii 24\}$        &                &   &                  &$\bullet$         & $\bullet$        &                  &$\bullet$         & $\bullet$        & $\bullet$        & $\bullet$        & $\bullet$        &                  &                  &                   \\
\hline                                                                                                                                                                                                                                                                                                       $\{=_p,=_i, <, \ip3, \ip4, \ii 44, \ii 04, \ii 24\}$&                 &  &                  &$\bullet$         & $\bullet$        &$\bullet$         &                  &                  & $\bullet$        &  $\bullet$       & $\bullet$        &                  &                  &                   \\
\hline
$\{=_p,=_i, <, \ip3, \ip4, \ii03\}$                 &                  & &                  &$\bullet$         & $\bullet$        &$\bullet$         &                  &                  & $\bullet$        & $\bullet$        &                  & $\bullet$        & $\bullet$        & $\bullet$         \\
\hline
$\{=_p, \ip4\}\cup\alli^{+}$                        &                  & &                  &$\bullet$         &                  &$\bullet$         &                  &                  &                  &                  &                  &                  &                  &                   \\
\hline
$\{=_p, < \}\cup\alli^{+}$                          &                   & &                  & $\bullet$       & $\bullet$        &$\bullet$         & $\bullet$                 &  $\bullet$       &                  &                  &                  &                  &                  &                   \\
\hline
$\{=_p, =_i, <, \ip 0, \ip 1, \ii 04\}$             &                  & &                  &                  &                  &                  &                  &                  &                  &                  &                  & $\bullet$        &                  & $\bullet$        \\
\hline
$\{=_p, =_i, <, \ip 0, \ip 1, \ii 44\}$             &                   & &                  &                  &                  &                  &                  &                  &                  &                  &                  & $\bullet$        & $\bullet$        &                   \\
\hline
$\{=_p, =_i, \ip 2, \ii 04, \ii 44\}$               &                   & &                  &                  &                  &                  &                 &                  &                  &                  &                  & $\bullet$        &                  &                   \\
\hline
$\{=_p, =_i, <, \ii 04, \ii 44\}$                   &                   & &                  &                  &                  &                  &                  &                  &                  &                  &                  & $\bullet$        &                  &                   \\
\hline
$\{=_p, =_i, <, \ip 2, \ii 24, \ii 44\}$                   &             &      &                  &                  &                  &                  &                  &                  &                  &                  &                  &                  & $\bullet$        &                   \\
\hline
$\{=_p, =_i, <, \ip 0, \ip 1, \ii 24\}$             &                  &  &                  &                  &                  &                  &                  &                  &                  &                  &                  &                  & $\bullet$        & $\bullet$         \\
\hline
\\
{\em Symmetric}                                     & $=_p$             & $=_i$              & $<$              & $\ip0$           & $\ip1$           & $\ip2$           & $\ip3$           & $\ip4$           & $\ii34$          & $\ii14$          & $\ii03$          & $\ii24$          & $\ii04$          & $\ii44$           \\
\hline
$\{=_p, <, \ip 3,\ip 4, \ii 44, \ii 04, \ii 24\}$                       &                   &$\bullet$           &                 &                  &                  &                  &                  &                  &                  &                  &                  &                  &                  &                   \\
\hline
$\{=_p, =_i, \ip 3, \ii 03\}$                       &                   & &$\bullet$         &                  &                  &                  &                  & $\bullet$        &                  &                  &                  &                  &                  &                   \\
\hline
$\{=_p, \ip 3, \ii 44, \ii 04, \ii 24\}$            &                   & &$\bullet$         &                  &                  &                  &                  & $\bullet$        &                  &                  &                  &                  &                  &                   \\
\hline
$\{=_p,=_i, <, \ip0, \ip1, \ii 44, \ii 04, \ii 24\}$&                   & &                  &                  &                  &$\bullet$         & $\bullet$        & $\bullet$        & $\bullet$        & $\bullet$        & $\bullet$         &                  &                  &                   \\
\hline
$\{=_p,=_i, <, \ip0, \ip1, \ii14\}$                 &                   & &                  &                  &                  &$\bullet$         & $\bullet$        & $\bullet$        & $\bullet$        &                  & $\bullet$        & $\bullet$        & $\bullet$        & $\bullet$         \\
\hline
$\{=_p, \ip0\}\cup\alli^{+}$                        &                   & &                  &                  &                  &$\bullet$         &                  & $\bullet$        &                  &                  &                  &                  &                  &                   \\
\hline
$\{=_p, =_i, <, \ip 3, \ip 4, \ii 04\}$             &                   & &                  &                  &                  &                  &                  &                  &                  &                  &                  & $\bullet$        &                  & $\bullet$         \\
\hline
$\{=_p, =_i, <, \ip 3, \ip 4, \ii 44\}$             &                   & &                  &                  &                  &                  &                  &                  &                  &                  &                  & $\bullet$        & $\bullet$        &                   \\
\hline
$\{=_p, =_i, <, \ip 3, \ip 4, \ii 24\}$             &                   & &                  &                  &                  &                  &                  &                  &                  &                  &                  &                  & $\bullet$        & $\bullet$         \\
\hline
\end{tabular}                                                                                                                                                                                                                                         \caption{$\mathsf{MIS}(r)$, for each $r\in\allr^+$; upper part: sets for which we give an explicit construction; lower part: symmetric ones. - Classes: $\Lin$, $\Dis$.}(Lemma~\ref{lem:tab:MISRlin}.)\label{tab:MISRlin}
\end{center}
\end{table}

\begin{lem}\label{lem:tab:MISRlin}
Tab.~\ref{tab:MISRlin} is correct and complete.
\end{lem}

\begin{proof}
We need to give a construction for 18 different sets. Most of these constructions are similar to each other, but
no part of any of them can be actually be re-used. Notice that every construction is based on a discrete set: this is why definability results over the classes $\Dis$ and $\Lin$ coincide. Observe, also, that most of the constructions are actually based on finite sets, which suggests that the
behaviour of these languages on $\Fin$ is very similar to that of $\Lin$ and $\Dis$; nevertheless, as we have mentioned, in Part II a deeper analysis
is required to complete the case $\Fin$.

\medskip

Let $S$ be \underline{$\{\ip 0, \ip 2, \ip 4, \}\cup\alli^{+}$}; we have to prove that it is $=_p$-incomplete. To this end, we simply take $\mathbb D = \mathbb D' =\{a<b\}$ and $\zeta=(\zeta_p,\zeta_i)$, where $\zeta_i=Id_i$, $\zeta_p=\{(a,a'),(a,b'),(b,a')\}$. The relation $\zeta:\mathcal F\rightarrow\mathcal F'$ is clearly total and surjective with respect to both interval and points, and every relations in $S$ is trivially respected. However, the pairs $(a,a')$ and $(a,b')$ show that equality between points is not respected. Let $S$ be \underline{$\{=_p, <, \ip 0, \ip 1, \ii 44, \ii 04, \ii 24\}$}. We have to prove here that it is $=_i$-incomplete. In this case, let $\mathbb D= \mathbb D'=\{a<b<c\}$,  $\zeta_p=Id_p$ and $\zeta_i=\{([a,c],[a,b])\}\cup Id_i$. Again, $S$ is total and surjective. Moreover, interval-interval relations are trivially respected, and point-interval relations are respected thanks to the fact that the {\em beginning} points of intervals are maintained through $\zeta$. But $[a,b]$ is the image of both $[a,b]$ and $[a,c]$, proving that $=_i$ is broken. A very similar construction, based on the same structures, with $\zeta_p=Id_p$ and $\zeta_i=\{([b,c],[a,b])\}\cup Id_i$, covers the $=_i$-incompleteness of \underline{$\{=_p, <, \ip 2, \ii 44, \ii 04, \ii 24\}$}. Let $S$ be \underline{$\{=_p, =_i, \ip 1, \ii 14\}$}; we shall prove that it is $\ip0$- and $<$-incomplete. In this case, let $\mathbb D=\mathbb D'=\mathbb Z$, and define $\zeta=(\zeta_p,\zeta_i)$ as follows: $(a,-a)\in\zeta_p$ for every $a\in\mathbb Z$, and $([a,b],[-a,-a+|b-a|])\in\zeta_i$ for every $[a,b]\in\mathbb I(\mathbb D)$, so that the length of every interval is preserved. Now, the interval-interval relation $\ii{1}{4}$ must be respected by $\zeta$, as intervals are mapped along with their beginning point, while the ending point does not move relative to the beginning point; for the same reason, $\ip 1$ is also respected. Finally, $\zeta$ is a bijection, so $=_i$ and $=_p$ are respected too; as it breaks both $<$ and $\ip 0$, these cannot be expressed in this language. Let $S$ be \underline{$\{=_p, \ip 1, \ii 44, \ii 04, \ii 24\}$}; we will prove that it is $<$- and $\ip 0$-incomplete. We take $\mathbb D=\mathbb D'=\{a<b<c\}$ and define $\zeta_p=\{(a,b'),(b, a'),(c,c')\}$ and $\zeta_i=\{([a,b],[b',c']),([b,c],[a',b']),([b,c],[a',c']),([a,c],[b',c'])\}$. All interval-interval relations are (vacuously) respected, $\ip 1$ is respected as well (as {\em beginning} points of intervals are preserved), and equality between points is respected too (as $\zeta_p$ is a bijection), but $<$ and $\ip 0$ are broken. If $S$ be \underline{$\{=_p, \ip 0, \ip 2, \ip 4\}\cup\alli^{+}$}, then take $\mathbb D=\mathbb D'=\{a<b\}$, $\zeta_i = \{([a,b], [a', b']) \}$ and $\zeta_p=\{(a,b'),(b,a')\}$, which respects $S$ and breaks $<$, $\ip 1$, and $\ip 3$, as it can be immediately verified. Now, let $S$ be \underline{$\{=_p, \ip 2\}\cup\alli^{+}$}. To prove that $S$ is $\ip 0$- and $\ip 4$-incomplete we take $\mathbb D=\mathbb D'=\{a<b<c\}$, and $\zeta=(\zeta_p,\zeta_i)$, where $\zeta_i=Id_i$, $\zeta_p=\{(a,c'),(c,a'),(b,b')\}$. The relation $\zeta:\mathcal F\rightarrow\mathcal F'$ is clearly surjective respectively to both interval and points, and every relations in $S$ is trivially preserved; on the other hand, both $\ip 0$ and $\ip 4$ are broken.
Let $S$ be \underline{$\{=_p,=_i,<,\ip 2,\ii 44, \ii 04, \ii 24\}$}; we prove that it is $\ip 0$-,$\ip1$-,$\ip 3$-,$\ip 4$-,$\ii 14$-,$\ip 03$-, and $\ii 34$-incomplete. It is enough to take $\mathbb D=\mathbb D'=\{a<b<c\}$, with $\zeta=(\zeta_p,\zeta_i)$ defined as $\zeta_p=Id_p$ and $\zeta_i=\{([a,c],[a',c']),([a,b],$ $[b',c']),([b,c],[a',b'])\}$, to have a relation $\zeta$ which is surjective respectively to both intervals and points, preserves every relations in $S$, and breaks all the indicated relations. If $S$ is \underline{$\{=_p,=_i, <, \ip3, \ip4, \ii 44, \ii 04, \ii 24\}$}, we have to prove that it is $\ip 0$-,$\ip1$-,$\ip 2$-,$\ii 14$-,$\ii 03$-, and $\ii 34$-incomplete. Take $\mathbb D=\mathbb D'=\{a<b<c\}$, $\zeta_p=Id_p$, and $\zeta_i=\{([b,c],[a',c']),([a,c],[b',c']),$ $([a,b],[a',b'])\}$. It is easy to verify that all 9 relations in $S$ are respected, but none of the indicated ones are. Now, let $S$ be \underline{$\{=_p,=_i, <, \ip3, \ip4, \ii03\}$}; we prove that it is $\ip0$-,$\ip1$-, and $\ip2$-incomplete, and, also, $r$-incomplete for every $r \in\alli^{+}\setminus\{\ii03,=_i\}$. As in Lemma~\ref{lem:tab:MISIlin}, we proceed in two steps, first proving the incompleteness for all indicated relations but $\ii 04$, and, then, dealing with $\ii 04$ on its own. As for the first step, define $\mathcal F$ and $\mathcal F'$ both based on the iterated discrete structure $\mathbb D^*$, as in Def.~\ref{def:iterated}. Define a relation $\zeta$ between them as follows:
\begin{enumerate}[label={(\emph{\roman*})}]
\item $\zeta_p=Id_p$; \item $\zeta_i$ is the union of
  \begin{align*}
    \{([(n,q),(m,r)],[(n',q'),(m',r')]) \mid q = r \mbox{ and } n, m \in  \mathbb{Z}\}
  \end{align*}
  and
  \begin{align*}
    \{ ([(n,q),(m,r)],[(n',q'-|r'-q'|),(m',r')]) \mid q < r \mbox{ and } n, m \in  \mathbb{Z}\}.
  \end{align*}
\end{enumerate} In this construction, which generalizes to the case of $\allr^+$ the one already seen in Lemma~\ref{lem:tab:MISIlin} by adding the $\zeta_p$ component, finite intervals (i.e., those containing a finite number of points) are related to themselves, while infinite ones are related to intervals with the same {\em ending} point but twice the length (w.r.t. the rational coordinate). We want to prove that $\zeta$ is an $S$-relation. Obviously, the only interesting cases are those that involve at least one infinite interval. Suppose, first, that $[(n,q),(m,r)]$ is $\ip 3$- or $\ip 4$-related to some point; under $\zeta$, these relations are clearly respected. So, consider two intervals $[(n,q),(m,r)]$ and $[(l,s),(m,r)]$ with $(l,s) \prec (n,q)$; they are $\ii{0}{3}$-related, so either both are infinite or $[(n,q),(m,r)]$ is finite and $[(l,s),(m,r)]$ is infinite, so under $\zeta$ the endpoints are kept constant and either both are `doubled' in length or $[(n,q),(m,r)]$ is kept fixed and only $[(l,s),(m,r)]$ is `doubled'. Either way, $\ii{0}{3}$ is respected. Since $\zeta$ is a bijection, equalities of both sort is respected too, and therefore $\zeta$ is a proper $S$-relation. It is clear, on the other hand, that $\ip 0,\ip1,\ip 2$ and $\ii{3}{4}, \ii{14}, \ii{2}{4}$ and $\ii{4}{4}$ are broken. As for the second step, define, again, $\mathcal F$ and $\mathcal F'$ both based on an iterated discrete structure $\mathbb D^*$, and define a new relation $\zeta$ as follows: \begin{enumerate*}[label={(\emph{\roman*})}] \item $\zeta_p=Id_p$; \item $\zeta_i$ is the union of $\{([(n,q),(m,r)],[(n',q'),(m',r')]) \mid q = r \mbox{ and } n, m \in  \mathbb{Z}\}$ and $\{ ([(n,q),(m,r)],[(n',q'+(|r'-q'|/2)),(m',r')]) \mid q < r \mbox{ and } n, m \in  \mathbb{Z}\}$\end{enumerate*}. In this way, finite intervals are related to themselves, while infinite ones are related to intervals with the same {\em ending} point but {\em half} the length (with respect to the rational coordinate). It is straightforward to check that this relation preserves both equalities and $\ii{0}{3}$, and that it breaks $\ii{0}{4}$ as we wanted. Let $S$ be \underline{$\{=_p, \ip4\}\cup\alli^{+}$}; we prove that it is $\ip 0$- and $\ip 2$-incomplete. For this, we take $\mathbb D=\mathbb D'=\{a<b<c\}$;  we define $\zeta_p=\{(a,b'),(b,a'),(c,c')\}$ and $\zeta_i=Id_i$, and we have that all interval-interval relations are respected, $\ip 4$ is respected as well ($\zeta_p(c,c')$ holds, and $c$ and $c'$ are the only points in relation $\ip 4$ with some interval), and equality between points is respected too (as $\zeta_p$ is a bijection), but $\ip 0$ and $\ip 2$ are broken. Let $S$ be \underline{$\{=_p, <\}\cup\alli^{+}$}; we prove here that this set is $\allm^+$-incomplete. To this aim we can simply take two copies of the integers, $\zeta_p=\{(n,n'+1)\}$ for every $n\in\mathbb Z$, and $\zeta_i=Id_i$; clearly, $\zeta$ respects $S$, as the only broken relations are the mixed ones. If $S$ is \underline{$\{=_p, =_i, <,\ip 0, \ip 1, \ii 04\}$}, for which we have to prove its $\ii 44$- and $\ii 24$-incompleteness, we take $\mathbb D=\mathbb D'=\{a<b<c<d\}$, $\zeta_p=Id_p$, and $\zeta_i=\{([a,b],[a',c']),([a,c],[a',b'])$ plus the identity relation on every other interval; relations $\ip 0$ and $\ip1$ are preserved, since the {\em starting} point of every interval does not change, the relation $\ii 04$ is preserved as well, as the only intervals that are not identically related have no {\em contained} intervals (neither they are {\em contained} by any other), and since $\zeta$ is a bijection, equalities are preserved too. Nevertheless, the two indicated relations are broken. Let $S$ be \underline{$\{=_p, =_i, <,\ip 0, \ip 1, \ii 44\}$}; we have to prove that it is $\ii 04$- and $\ii 24$-incomplete. We proceed as before, by taking $\mathbb D=\mathbb D=\{a<b<c<d\}$, $\zeta_p=Id_p$, and $\zeta_i=\{([a,c],[a',d']),([a,d],[a',c'])\}$ plus the identity relation on every other interval, and, as it is immediate to check, all interval-interval and point-interval relations in $S$ are respected, but $\ii 04$ and $\ii 24$ are not respected. When $S$ is \underline{$\{=_i, =_p, \ip 2, \ii04, \ii44\}$}, we need to prove that it is $\ii 24$-incomplete. To this end, we take $\mathbb D=\mathbb D'=\{a<b<c<d\}$, $\zeta_p(b)=c'$, $\zeta_p(c)=b'$ plus the identity relation over the other points, and $\zeta_i=\{([a,c],[b',d']),([b,d],[a',c'])$ plus the identity relation on every other interval.
The relation $\ip 2$ is preserved, since the internal points of every interval involved in $\zeta$ are moved along the interval that {\em contains} them, the relations $\ii 04$ and $\ii 44$ are respected, as the only intervals that are not identically related via $\zeta$ are not {\em contained} in, or {\em before} any other interval, and since $\zeta$ is a bijection, equalities are respected as well. Conversely, $\ii 24$ is broken.
An identical construction, where we simply take $\zeta_p=Id_p$, proves the $\ii 24$-incompleteness of \underline{$\{=_i, =_p, <, \ii04, \ii44\}$}.
When we have to prove that \underline{$\{=_i, =_p, <, \ip2, \ii24,\ii 44\}$} is $\ii 04$-incomplete, we simply take $\mathbb D=\mathbb D'=\{a<b<c<d\}$, $\zeta_p=Id_p$, and $\zeta_i=\{([b,c],[c',d']),([c,d],[b',c'])$ plus the identity relation on every other interval, and, finally,
to prove that \underline{$\{=_p, =_i, <,\ip 0, \ip 1, \ii 24\}$} is $\ii 44$- and $\ii 04$-incomplete, we take $\mathbb D=\mathbb D'=\{a<b<c<d\}$, $\zeta_p=Id_p$, and $\zeta_i=\{([a,b],[a',d']),([a,d],[a',b'])$ plus the identity relation on every other interval. \end{proof}


\section{Harvest: The Complete Picture for \texorpdfstring{$\Lin$}{Lin} and \texorpdfstring{$\Dis$}{Dis}}\label{sec:harvest}

\begin{figure}[t]
\scriptsize
\centering
\begin{code}{60mm}
\FUNCTION{MinDef\_MaxUndef}{}
\BEGIN
\FORALL S\subseteq\allr^{+}\\
\BEGIN
\IF \forall r\in\allr^{+} (r\in Closure(S)) \AND\\
    S\ is\ minimal, \ list\ S \ as\ \mathrm{mcs(\allr^{+})}\\
\IF \exists r\in\allr^{+} (r\notin Closure(S)) \AND\\
    S\ is\ maximal, \ list\ S \ as\ \mathrm{MIS(\allr^{+})}\\
\END\\
\RETURN;
\END \\ \\
\end{code}
\begin{code}{50mm}
\FUNCTION{IsExpressiveAs}{set\ S,S'\ }
\BEGIN
S=Closure(S);\\
S'=Closure(S');\\
\IF \ S=S' \RETURN \ Yes;\\
\RETURN\ No;
\END \\ \\
\end{code}
\caption{Pseudo-code to compare the expressive power of subsets of $\allr^{+}$.}\label{fig:diff}
\end{figure}

We are now capable to identify all expressively different subsets of $\allr^{+}$, $\alli^+$, and $\allm^+$. One can easily establish the expressive power of a particular subset $S$ with respect to another subset $S'$ by simply comparing, containment-wise, the closure (by definability) of $S$ with the closure of $S'$ (see Fig.~\ref{fig:diff}). Clearly, it is very difficult to display the resulting Hasse-diagram; we then limit ourselves to list all maximally incomplete sets for each case. The following theorem is a consequence of all results seen so far.

\begin{table}[t]
\small
\begin{center}
\begin{tabular}{|p{0.18\textwidth}|p{0.34\textwidth}|}
\hline
\multicolumn{2}{|c|}{$\allr^+$}\\
\hline
$\mathsf{mcs}$ & $\mathsf{MIS}$ \\
\hline
$\{\ip0, \ip2, < \}$                  &$\{=_i, =_p, <,\ip0, \ip1, \ii14\}$               \\
$\{\ip0, \ip3 \}$                     &$\{=_i, =_p, <,\ip0, \ip1, \ii24, \ii04, \ii44\}$ \\
$\{\ip0, \ip4, < \}$                  &$\{=_p, \ip0, \ip2, \ip4\}\cup\alli^{+}$          \\
$\{\ip0, \ii14, \ii24, < \}$          &$\{=_i, =_p, <,\ip2, \ii24, \ii04, \ii44\}$       \\
$\{\ip0, \ii14, \ii04, < \}$          &$\{=_p, <\}\cup\alli^{+}$                         \\
$\{\ip0, \ii14, \ii44, < \}$          &$\{=_i, =_p, <,\ip3, \ip4, \ii03\}$               \\
$\{\ip0, \ii03, < \}$                 &$\{=_i, =_p, <,\ip3, \ip4, \ii24, \ii04, \ii44\}$ \\
$\{\ip0, \ii34, < \}$                 &                                                  \\
$\{\ip1, \ip2 \}$                     &                                                  \\
$\{\ip1, \ip3 \}$                     &                                                  \\
$\{\ip1, \ip4 \}$                     &                                                  \\
$\{\ip1, \ii14, \ii24 \}$             &                                                  \\
$\{\ip1, \ii14, \ii04 \}$             &                                                  \\
$\{\ip1, \ii14, \ii44 \}$             &                                                  \\
$\{\ip1, \ii03 \}$                    &                                                  \\
$\{\ip1, \ii34 \}$                    &                                                  \\
$\{\ip2, \ip3 \}$                     &                                                  \\
$\{\ip2, \ip4, < \}$                  &                                                  \\
$\{\ip2, \ii14, < \}$                 &                                                  \\
$\{\ip2, \ii03, < \}$                 &                                                  \\
$\{\ip2, \ii34, < \}$                 &                                                  \\
$\{\ip3, \ii14 \}$                    &                                                  \\
$\{\ip3, \ii24, \ii03 \}$             &                                                  \\
$\{\ip3, \ii03, \ii04 \}$             &                                                  \\
$\{\ip3, \ii03, \ii44 \}$             &                                                  \\
$\{\ip3, \ii34 \}$                    &                                                  \\
$\{\ip4, \ii14, < \}$                 &                                                  \\
$\{\ip4, \ii24, \ii03, < \}$          &                                                  \\
$\{\ip4, \ii03, \ii04, < \}$          &                                                  \\
$\{\ip4, \ii03, \ii44, < \}$          &                                                  \\
$\{\ip4, \ii34, < \}$                 &                                                  \\
\hline
\end{tabular}
\begin{tabular}{|p{0.16\textwidth}|p{0.20\textwidth}|}
\hline
\multicolumn{2}{|c|}{$\alli^+$}\\
\hline
$\mathsf{mcs}$ & $\mathsf{MIS}$ \\
\hline
$\{\ii34\}$       &$\{=_i,\ii14\}$               \\
$\{\ii14,\ii24\}$ &$\{=_i,\ii24, \ii04, \ii44\}$ \\
$\{\ii14,\ii04\}$ &$\{=_i,\ii03\}$               \\
$\{\ii14,\ii03\}$ &                              \\
$\{\ii14,\ii44\}$ &                              \\
$\{\ii03,\ii24\}$ &                              \\
$\{\ii03,\ii04\}$ &                              \\
$\{\ii03,\ii44\}$ &                              \\
\hline
\multicolumn{2}{c}{}\\
\hline
\multicolumn{2}{|c|}{$\allm^+$}\\
\hline
$\mathsf{mcs}$ & $\mathsf{MIS}$ \\
\hline
$\{\ip 0,\ip 3\}$ &$\{\ip 0,\ip 1\}$      \\
$\{\ip 1,\ip 2\}$ &$\{\ip0,\ip2,\ip 4\}$  \\
$\{\ip 1,\ip 3\}$ &$\{\ip 3,\ip 4\}$      \\
$\{\ip 1,\ip 4\}$ &                              \\
$\{\ip 2,\ip 3\}$ &                              \\
\hline
\end{tabular}
\caption{Left: minimally $\allr^+$-complete and maximally $\allr^+$-incomplete sets.
Right, top: minimally $\alli^+$-complete and maximally $\alli^+$-incomplete sets. Right, bottom: minimally $\allm^+$-complete and maximally $\allm^+$-incomplete sets. - Classes: $\Lin$, $\Dis$.}\label{tab:harvest}
\end{center}
\end{table}

\begin{thm}
If a set of relations is listed:
\begin{itemize}[nosep]
\item as $\mathsf{mcs}(\allr^{+})$ (resp., $\mathsf{mcs}(\alli^{+})$, $\mathsf{mcs}(\allm^{+})$) in Tab.~\ref{tab:harvest} left (resp., right-top, right-bottom) side, left column, then it is minimally $\allr^{+}$-complete (resp., minimally $\alli^{+}$-complete, minimally $\allm^{+}$-complete) in the class of all linearly ordered sets and in the class of all discrete linearly ordered sets.
\item as $\mathsf{MIS}(\allr^{+})$ (resp., $\mathsf{MIS}(\alli^{+})$, $\mathsf{MIS}(\allm^{+})$) in Tab.~\ref{tab:harvest} left (resp., right-top, right-bottom) side, right column, then it is maximally $\allr^{+}$-incomplete (resp., maximally $\alli^{+}$-incomplete, maximally $\allm^{+}$-incomplete) in the class of all linearly ordered sets and in the class of all discrete linearly ordered sets.
\end{itemize}
\end{thm}

\section{Conclusions}

We considered here the two-sorted first-order temporal language that includes relations between intervals,
points, and inter-sort, and we treated equality between points and between intervals as any other relation, with no special role. Under four
different assumptions on the underlying structure, namely, linearity only, linearity+discreteness, linearity+density, and linearity+unboundedness, we
asked the question: which relation can be first-order defined by which subset of all relations? As a result, we identified all possible inter-definability
between relations, all minimally complete, and all maximally incomplete subsets of relations. These inter-definability results allow one to effectively compute all expressively different subsets of relations, and, with minimal effort, also all expressively different subsets of relations for the interesting
sub-languages of interval relations only or mixed relations only. Two out of four interesting classes of linearly ordered sets are treated in this Part I, and the remaining two are dealt with in Part II (forthcoming). There are several aspects of temporal reasoning in computer science to which this
extensive study can be related:
\begin{itemize}[nosep]
\item first-order logic over linear orders extended with temporal relations between points, intervals and mixed, are the very foundation of modal logics
for temporal reasoning, and it is necessary to have a complete understanding of the former in order to deal with the latter;
\item automated reasoning techniques for interval-based modal logics are at their first stages; an uncommon, but promising approach is to treat them as pure
modal logics over particular Kripke-frames, whose first-order properties are, in fact, representation theorems such as those (indirectly) treated in this paper. As a future work, we also plan to systematically study the area of representation theorems;
\item the decidability of pure first-order theories extended with interval relations is well-known \cite{lad87}; nevertheless, these results hinge on the decidability of MFO[$<$], while we believe that they could be refined both algorithmically and computationally;
\item the study of other related languages, important in artificial intelligence, can benefit from our results, such as first-order and modal logics for
spatial reasoning where basic objects are, for example, rectangles.
\end{itemize}

\bibliographystyle{alpha}

\bibliography{biblio}

\newpage

\section*{Appendix: Hasse Diagrams}

A pictorial representation of the results of this paper, in terms of relative expressive power of subsets of relations is shown here.

\begin{figure}
\begin{center}
    \begin{tikzpicture}
    [standard/.style={circle,fill=black!100, inner sep=0pt, minimum size=0mm}, scale=1.3]

    \node (top)   at (0, 2) [draw=none,fill=none]  {$\allm^{+}$};
    \node (01)   at  (-3, 1) [draw=none,fill=none]  {$\{\ip{0}, \ip{1} \}$};
    \node (34)   at  (3, 1) [draw=none,fill=none]  {$\{\ip{3}, \ip{4} \}$};
    \node (024)   at  (0, 1) [draw=none,fill=none]  {$\{\ip{0}, \ip{2}, \ip{4} \}$};
    \node (0)   at  (-2.5, 0) [draw=none,fill=none]  {$\{\ip{0} \}$};
    \node (1)   at  (-3.5, 0) [draw=none,fill=none]  {$\{\ip{1} \}$};
    \node (3)   at  (3.5, 0) [draw=none,fill=none]  {$\{\ip{3} \}$};
    \node (4)   at  (2.5, 0) [draw=none,fill=none]  {$\{\ip{4} \}$};
    \node (2)   at  (0, 0) [draw=none,fill=none]  {$\{\ip{2} \}$};
    \node (empty) at  (0, -1) [draw=none,fill=none]  {$\varnothing $};

    \draw[-] (top) to node[swap] {} (01);
    \draw[-] (top) to node[swap] {} (34);
    \draw[-] (top) to node[swap] {} (024);
    \draw[-] (01) to node[swap] {} (0);
    \draw[-] (01) to node[swap] {} (1);
    \draw[-] (34) to node[swap] {} (3);
    \draw[-] (34) to node[swap] {} (4);
    \draw[-] (024) to node[swap] {} (0);
    \draw[-] (024) to node[swap] {} (2);
    \draw[-] (024) to node[swap] {} (4);
    \draw[-] (0) to node[swap] {} (empty);
    \draw[-] (1) to node[swap] {} (empty);
    \draw[-] (2) to node[swap] {} (empty);
    \draw[-] (3) to node[swap] {} (empty);
    \draw[-] (4) to node[swap] {} (empty);

    \end{tikzpicture}
    \end{center}\caption{The lattice of expressively different fragments of $FO+\allm^{+}$}\label{Fig:RelExpress:Mixed}
\end{figure}

\begin{figure}
\begin{center}
    \begin{tikzpicture}
    [standard/.style={circle,fill=black!100, inner sep=0pt, minimum size=0mm}, scale=1.3]
    \node (m)   at (5, 2) [draw=none,fill=none]  {\footnotesize{$\alli^{+}$}};
    \node (s)   at (5, 1) [draw=none,fill=none] {$\{ \ii{1}{4}, =_i \}$};
    \node (f)   at (6.5, 1) [draw=none,fill=none] {$\{ \ii{0}{3}, =_i \}$};

    \node (bdoe) at (2.0, 1) [draw=none,fill=none] {\footnotesize{$\{ \ii{4}{4}, \ii{0}{4}, \ii{2}{4}, =_{i} \}$}};
    \node (bdo)  at (-1.2, 0) [draw=none,fill=none] {$\{ \ii{4}{4}, \ii{0}{4}, \ii{2}{4}\}$};
    \node (bde)  at ( 0.7, 0) [draw=none,fill=none] {$\{ \ii{4}{4}, \ii{0}{4}, =_{i}\}$};
    \node (boe)  at ( 2.4, 0) [draw=none,fill=none] {$\{ \ii{4}{4}, \ii{2}{4}, =_{i} \}$};
    \node (doe)  at ( 4.1, 0) [draw=none,fill=none] {$\{ \ii{0}{4}, \ii{2}{4}, =_{i} \}$};
    \node (bd)  at ( -2.1, -1.5) [draw=none,fill=none] {$\{ \ii{4}{4}, \ii{0}{4} \}$};
    \node (bo)  at ( -0.8, -1.5) [draw=none,fill=none] {$\{ \ii{4}{4}, \ii{2}{4} \}$};
    \node (be)  at ( 0.5, -1.5) [draw=none,fill=none] {$\{ \ii{4}{4}, =_{i} \}$};
    \node (do)  at ( 1.7, -1.5) [draw=none,fill=none] {$\{ \ii{0}{4}, \ii{2}{4} \}$};
    \node (de)  at ( 3.0, -1.5) [draw=none,fill=none] {$\{ \ii{0}{4}, =_{i} \}$};
    \node (oe)  at ( 4.3, -1.5) [draw=none,fill=none] {$\{ \ii{2}{4}, =_{i} \}$};

    \node (b)  at (1 , -3) [draw=none,fill=none] {$\{ \ii{4}{4} \}$};
    \node (d)  at (2 , -3) [draw=none,fill=none] {$\{ \ii{0}{4} \}$};
    \node (o)  at (3 , -3) [draw=none,fill=none] {$\{ \ii{2}{4} \}$};
    \node (e)  at (6.5 , -3) [draw=none,fill=none] {$\{  =_{i}\}$};
    \node (empty) at (5, -4) [draw=none,fill=none] {$\emptyset$};

    \draw[-] (m) to node[swap] {} (f);
    \draw[-] (m) to node[swap] {} (s);
    \draw[-] (m) to node[swap] {} (bdoe);

    \draw[-] (bdoe) to node[swap] {} (bdo);
    \draw[-] (bdoe) to node[swap] {} (bde);
    \draw[-] (bdoe) to node[swap] {} (boe);
    \draw[-] (bdoe) to node[swap] {} (doe);

    \draw[-] (bdo) to node[swap] {} (bd);
    \draw[-] (bdo) to node[swap] {} (bo);
    \draw[-] (bdo) to node[swap] {} (do);

    \draw[-] (bde) to node[swap] {} (bd);
    \draw[-] (bde) to node[swap] {} (be);
    \draw[-] (bde) to node[swap] {} (de);

    \draw[-] (boe) to node[swap] {} (bo);
    \draw[-] (boe) to node[swap] {} (be);
    \draw[-] (boe) to node[swap] {} (oe);

    \draw[-] (doe) to node[swap] {} (do);
    \draw[-] (doe) to node[swap] {} (de);
    \draw[-] (doe) to node[swap] {} (oe);

    \draw[-] (bd) to node[swap] {} (b);
    \draw[-] (bd) to node[swap] {} (d);
    \draw[-] (be) to node[swap] {} (b);
    \draw[-] (be) to node[swap] {} (e);
    \draw[-] (oe) to node[swap] {} (o);
    \draw[-] (oe) to node[swap] {} (e);
    \draw[-] (de) to node[swap] {} (d);
    \draw[-] (de) to node[swap] {} (e);
    \draw[-] (do) to node[swap] {} (d);
    \draw[-] (do) to node[swap] {} (o);
    \draw[-] (bo) to node[swap] {} (b);
    \draw[-] (bo) to node[swap] {} (o);

    \draw[-] (s) to node[swap] {} (e);
    \draw[-] (f) to node[swap] {} (e);

    \draw[-] (b) to node[swap] {} (empty);
    \draw[-] (d) to node[swap] {} (empty);
    \draw[-] (o) to node[swap] {} (empty);
    \draw[-] (e) to node[swap] {} (empty);

    \end{tikzpicture}
    \end{center}\caption{The lattice of expressively different fragments of $FO+\alli^{+}$}\label{Fig:RelExpress:All}
\end{figure}

\begin{figure}
\begin{center}
\includegraphics[width=15cm,height=21.5cm]{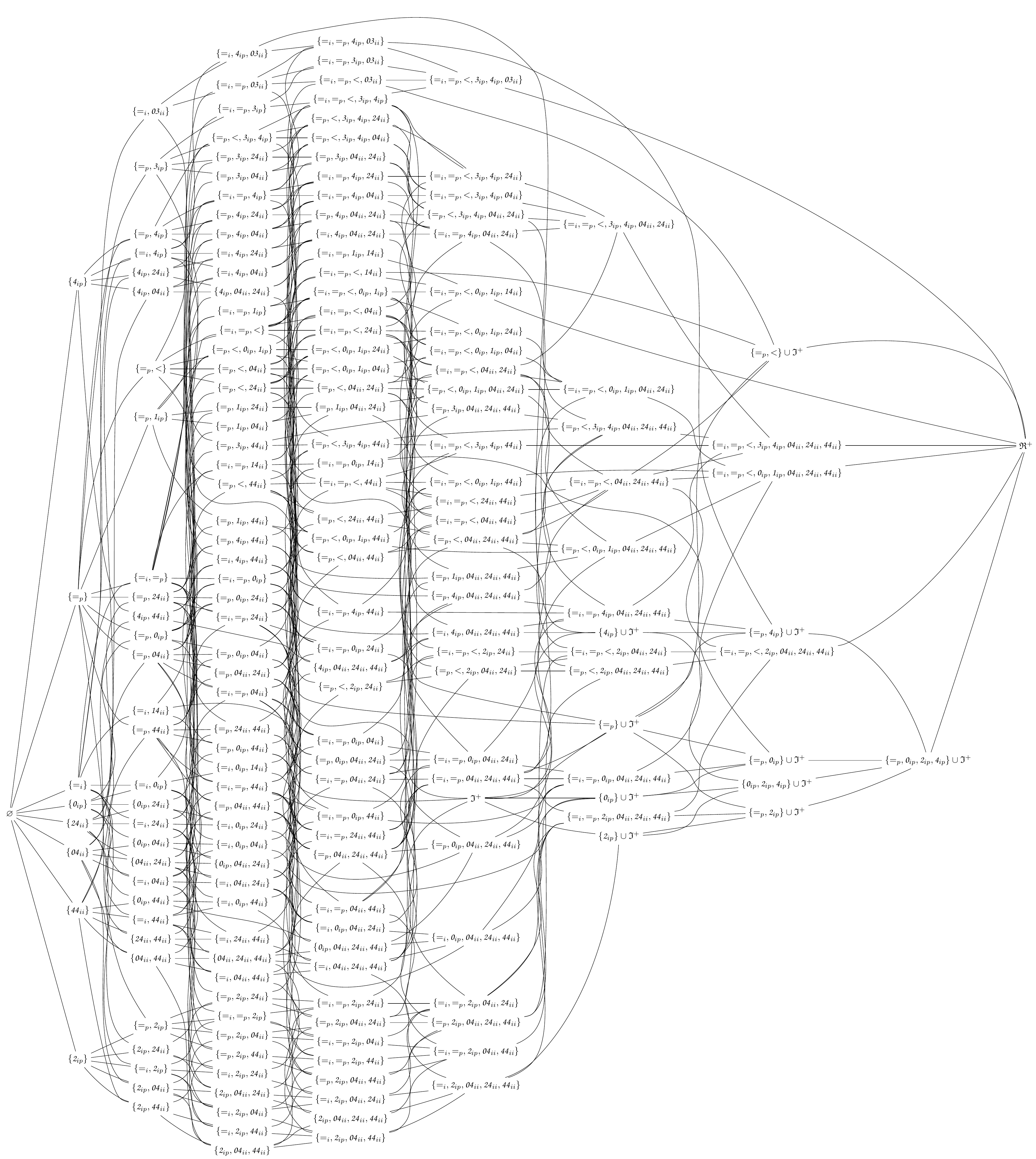}
\end{center}\caption{The lattice of expressively different fragments of $FO+\allr^{+}$}\label{Fig:RelExpress:Mixed}
\end{figure}

\end{document}